\documentclass[a4paper,11pt,reqno]{amsart}
\usepackage[utf8]{inputenc}
\usepackage{mathrsfs}
\usepackage{dsfont}
\usepackage{hyperref}
\usepackage{amsmath}
\usepackage{amssymb}
\usepackage{amsthm}
\usepackage{amsfonts}
\usepackage{amstext}
\usepackage{amsopn}
\usepackage{amsxtra}
\usepackage{mathrsfs}
\usepackage{dsfont}
\usepackage{esint}
\usepackage{enumitem}
\usepackage{graphicx}

\newtheorem{theorem}{Theorem}
\newtheorem{lemma}[theorem]{Lemma}

\newtheorem{corollary}[theorem]{Corollary}
\newtheorem{remark}[theorem]{Remark}

\newcommand{\R}{\mathbb{R}}

\newcommand{\C}{\mathbb{C}}
\newcommand{\N}{\mathbb{N}}
\newcommand{\bP}{\mathbb{P}}
\newcommand{\bS}{\mathbb{S}}
\newcommand{\bE}{\mathbb{E}}

\newcommand{\ii}{\infty}
\newcommand\1{{\ensuremath {\mathds 1} }}
\renewcommand\phi{\varphi}

\newcommand{\cT}{\mathcal{T}}

\newcommand{\cB}{\mathcal{B}}

\newcommand{\cK}{\mathcal{K}}
\newcommand{\cE}{\mathcal{E}}
\newcommand{\cF}{\mathcal{F}}
\newcommand{\cN}{\mathcal{N}}

\renewcommand{\geq}{\geqslant}
\renewcommand{\leq}{\leqslant}

\renewcommand{\tilde}{\widetilde}
\newcommand{\eps}{\varepsilon}
\usepackage{color}

\newcommand{\nn}{\nonumber}

\newcommand{\rd}{\mathrm{d}}
\newcommand{\dx}{\rd x}
\newcommand{\dy}{\rd y}

\title[Improved Lieb-Oxford bound]{Improved Lieb-Oxford bound on the indirect and exchange energies}

\author[M. Lewin]{Mathieu Lewin}
\address{CNRS \& CEREMADE, Universit\'e Paris-Dauphine, PSL University, 75016 Paris, France}
\email{mathieu.lewin@math.cnrs.fr}

\author[E.H. Lieb]{Elliott H. Lieb}
\address{Departments of Mathematics and Physics, Jadwin Hall, Princeton University, Washington Rd., Princeton, NJ 08544, USA}
\email{lieb@princeton.edu}

\author[R. Seiringer]{Robert Seiringer}
\address{IST Austria (Institute of Science and Technology Austria), Am Campus 1, 3400 Klosterneuburg, Austria}
\email{robert.seiringer@ist.ac.at}

\date{\today. This article belongs to the themed collection: \emph{Mathematical Physics and Numerical Simulation of Many-Particle Systems}; V. Bach and L. Delle Site (eds.).}

\begin{document}

\begin{abstract}
The Lieb-Oxford inequality provides a lower bound on the Coulomb energy of a classical system of $N$ identical charges only in terms of their one-particle density. We prove here a new estimate on the best constant in this inequality. Numerical evaluation provides the value 1.58, which is a significant improvement to the previously known value 1.64. The best constant has recently been shown to be larger than 1.44. In a second part, we prove that the constant can be reduced to 1.25 when the inequality is restricted to Hartree-Fock states. This is the first proof that the exchange term is always much lower than the full indirect Coulomb energy.

\bigskip

\noindent \sl \copyright~2022 by the authors. This paper may be reproduced, in its entirety, for non-commercial purposes.
\end{abstract}

\maketitle

\tableofcontents

\section{Introduction and main result}

Density Functional Theory (DFT)~\cite{DreGro-90,ParYan-94,EngDre-11,BurWag-13,LewLieSei-19_ppt} has become the method of choice to simulate large molecules and solids. The hundreds of density functionals currently available can essentially be classified in two categories, the \emph{semi-empirical} and \emph{non-empirical} functionals~\cite{MarHea-17}. The first class contains functionals whose parameters are fitted to accurate reference values~\cite{Becke-93,Becke-97}. On the other hand, the form of the non-empirical functionals is chosen based on physical considerations, but the unknown parameters are all fixed to ensure that the functional satisfies a list of exactly known properties. This strategy has been used for instance for the famous PBE~\cite{PerBurErn-96} and SCAN~\cite{SunRuzPer-15} functionals, which satisfy 11 and 17 exact constraints, respectively.

The Lieb-Oxford bound~\cite{LieOxf-80} provides an estimate on the lowest possible Coulomb energy of $N$ identical particles in terms of their charge density~$\rho$. This bound was used as a constraint for some density functionals, including PBE and SCAN~\cite{Perdew-91,LevPer-93,PerBurErn-96,TaoPerStaScu-03,SunPerRuz-15,SunRuzPer-15,Perdew_etal-16,PerSun-22}. For applications to DFT it is therefore important to improve the numerical value of the constant in the Lieb-Oxford inequality. The best value known at the moment is 1.64. It was derived by Chan and Handy~\cite{ChaHan-99}, based on a numerical optimization of one part of the original Lieb-Oxford argument~\cite{LieOxf-80}, which itself had provided the constant 1.68. It has recently been shown~\cite{CotPet-19b,LewLieSei-19b} that the best constant has to be larger than 1.44. In this paper, we prove a new estimate on the best constant. Numerical evaluation of the resulting expression then provides the approximate value 1.58.

The \emph{exchange energy}, which is obtained when restricting the problem to Hartree-Fock states, is also very useful for DFT. It was argued in~\cite{PerRuzSunBur-14,PerSun-22} that there should exist a Lieb-Oxford bound with a much lower constant for such states. In this article we shall prove that, indeed, the bound holds for the exchange energy with the constant 1.25, which is strictly less than the best lower bound 1.44 known in the general case.

\subsubsection*{\textbf{The indirect Coulomb energy}}
Before stating our main result, we start by defining the indirect energy and reviewing some of its properties. We consider a symmetric probability distribution $\bP$ over $(\R^3)^N$. The latter represents a system of $N$ classical or quantum identical particles. For electrons we have
\begin{equation}
 \bP(x_1,\dots,x_N)=\sum_{\sigma_1,\dots,\sigma_N\in\{\uparrow,\downarrow\}}|\Psi(x_1,\sigma_1,\dots,x_N,\sigma_N)|^2
 \label{eq:bP_Psi}
\end{equation}
where $\Psi$ is the corresponding quantum wave function and $\sigma_1,\dots,\sigma_N$ are the spin variables. A similar formula holds for density matrices. The associated one-particle charge density is defined by
$$\rho_\bP(x)=N\int_{(\R^3)^{N-1}}\rd\bP(x,x_2,\dots,x_N).$$
The \emph{indirect energy of $\bP$} is by definition the difference between the $\bP$--expectation of the $N$-body Coulomb interaction and the classical Coulomb energy of the charge density $\rho_\bP$ (also called the \emph{direct} or \emph{Hartree} or \emph{Coulomb term}):
\begin{multline}
\cE_\text{ind}[\bP]:=\int_{(\R^3)^N}\bigg(\sum_{1\leq j<k\leq N}\frac1{|x_j-x_k|}\bigg)\rd\bP(x_1,\dots,x_N)\\-\frac12\iint_{\R^3\times\R^3}\frac{\rho_\bP(x)\rho_\bP(y)}{|x-y|}\rd x\,\rd y.
\label{eq:cE_ind_bP}
\end{multline}
The \emph{indirect energy of a density $\rho$} is the smallest indirect energy one can reach with $N$-particle probabilities $\bP$ of density $\rho_\bP=\rho$:
\begin{equation}
E_\text{ind}[\rho]:=\inf\big\{\cE_\text{ind}[\bP]\ :\ \rho_\bP=\rho\big\}.
\label{eq:E_ind}
\end{equation}
This is well defined under the sole assumption that $\rho$ is a non-negative measure with $\int_{\R^3}\rho=N$, such that the direct term is finite. For instance the uniform measure of a sphere is allowed, but not a Dirac delta at a point.  In fact, we have the estimate
\begin{equation}
-\frac12\iint_{\R^3\times\R^3}\frac{\rho(x)\rho(y)}{|x-y|}\rd x\,\rd y\leq E_\text{ind}[\rho]\leq -\frac1{2N}\iint_{\R^3\times\R^3}\frac{\rho(x)\rho(y)}{|x-y|}\rd x\,\rd y\leq0.
\label{eq:bad_bound}
\end{equation}
The lower bound follows from the positivity of the Coulomb interaction and the upper bound is obtained after inserting the trial state $\bP(x_1,\dots,x_N)=N^{-N}\prod_{j=1}^N\rho(x_j)$ corresponding to i.i.d.~particles.

Under the same assumptions as above on $\rho$, the infimum in~\eqref{eq:E_ind} is known to be attained, that is, there exists an optimal $\bP$. This minimizer is typically not absolutely continuous with respect to the Lebesgue measure, and concentrates on a manifold of small dimension. This is an expression of the fact that the particles are believed to be ``strictly correlated''~\cite{Seidl-99,SeiPerLev-99,SeiGorSav-07,GorSeiVig-09,RasSeiGor-11,SeiBenKooGor-22,FriGerGor-22_ppt}, that is, the locations of the $N$ particles are completely determined by those of only a few of them, like in a crystal. It has also been proved using methods of multi-marginal optimal transportation~\cite{Kellerer-84,ButPasGor-12,Pascale-15,ButChaPas-18,MarGerNen-17} that there exists a \emph{dual potential} $V$. The latter is such that $\bP$ concentrates on the set of ground states for the classical $N$-particle Hamiltonian with external potential $V$. This is the classical equivalent of the Kohn-Sham potential in DFT.

In Density Functional Theory, it is often useful to rely on \emph{local} functionals. Due to the homogeneity of the Coulomb potential, the most natural local approximation to $E_\text{ind}[\rho]$ is a multiple of $\int_{\R^3}\rho(x)^{4/3}\,\rd x$. In fact, for densities which are very flat (slowly varying), the \emph{Local Density Approximation} predicts that $E_\text{ind}[\rho]$ converges to such a local term. This was recently proved in~\cite{LewLieSei-18}, where it is shown that
\begin{equation}
E_{\rm ind}[\rho(\cdot/N^{\frac13})]\underset{N\to\ii}\sim e_{\rm UEG}\int_{\R^3}\rho(x/N^{\frac13})^{\frac43}\,\rd x=N\,e_{\rm UEG}\int_{\R^3}\rho(x)^{\frac43}\,\rd x
\label{eq:LDA}
\end{equation}
for any fixed smooth density $\rho$ of integer mass. Here $e_{\rm UEG}$ is the energy per unit volume of the \emph{Uniform Electron Gas (UEG)} at unit density. Wigner predicted that this system is crystallized on the BCC lattice~\cite{Wigner-34}, which leads to the conjecture that
$$e_{\rm UEG}\overset{?}{=}\zeta_{\rm BCC}(1)\simeq -1.4442$$
where $\zeta$ is the Epstein Zeta function of this lattice~\cite{BlaLew-15,Lewin-22}. The constraint that the UEG must have a uniform density is not so easy to handle and the upper bound
\begin{equation}
e_{\rm UEG}\leq\zeta_{\rm BCC}(1)
\label{eq:UEG_upper}
\end{equation}
was surprisingly hard to prove. This was achieved only recently in~\cite{CotPet-19b,LewLieSei-19b}. A simple floating Wigner crystal does not yield the right answer~\cite{LewLie-15} and it is necessary to make it float in a layer of perfect fluid to damp the large boundary charge fluctuations~\cite{LewLieSei-19b}. A universal bound on the difference
$E_{\rm ind}[\rho]- e_{\rm UEG}\int_{\R^3}\rho(x)^{4/3}\,\rd x$ involving gradients corrections was recently derived in~\cite[App.~A]{LewLieSei-19} but it requires using the grand canonical version of $E_{\rm ind}[\rho]$~\cite{LewLieSei-19_ppt} which will not be discussed here.

\subsubsection*{\textbf{The Lieb-Oxford (LO) bound}}
It is natural to ask whether there exists a universal bound (valid for all possible densities) on $E_\text{ind}[\rho]$, involving only its local approximation $\int_{\R^3}\rho(x)^{4/3}\,\rd x$. The first proof of the existence of such a lower bound is due to Lieb and Oxford in~\cite{Lieb-79,LieOxf-80}. The version we obtain in this paper is partly based on numerical optimization and can be stated as
\begin{equation}
\boxed{E_\text{ind}[\rho]\geq-1.58\int_{\R^3}\rho(x)^{\frac43}\,\rd x.}
\label{eq:LO}
\end{equation}
The same inequality was shown first by Lieb~\cite{Lieb-79} with the constant 8.52, and then improved by Lieb-Oxford~\cite{LieOxf-80} to the much better 1.68. By numerically optimizing one part of the argument of Lieb-Oxford, Chan and Handy~\cite{ChaHan-99} managed to push down the constant to 1.64. This was the best known value prior to our work. Due to~\eqref{eq:LDA} and~\eqref{eq:UEG_upper}, it is known that the best constant is above 1.44. In fact, it was conjectured in~\cite{LevPer-93,OdaCap-07,RasPitCapPro-09} that the best Lieb-Oxford constant might be $-e_{\rm UEG}$. Hence~\eqref{eq:LO} is a rather significant decrease of 30\;\% of the gap to the conjectured value.

Let us mention that the constant in~\eqref{eq:LO} can be replaced by 1.451 when $\rho$ is a characteristic function, or at the expense of adding gradient corrections~\cite{LieNar-75,BenBleLos-12,LewLie-15}. Let us also recall that there cannot exist a negative \emph{upper} bound on $E_{\rm ind}[\rho]$ only in terms of $\int_{\R^3}\rho^{4/3}(x)\,\rd x$, for general densities. In fact, such a bound already fails at $N=1$ since the direct term cannot be estimated from below by $\int_{\R^3}\rho^{4/3}(x)\,\rd x$ for all $\rho$.

In order to show~\eqref{eq:LO} we first revisit in Section~\ref{sec:theory} the Lieb-Oxford proof and provide a new estimate on the best constant $c_{\rm LO}$. It takes the general form
\begin{equation}
c_{\rm LO}\leq \inf_{\mu,\nu}\cK(\mu,\nu),
\label{eq:general_form_proof}
\end{equation}
where the infimum is over all radial probability measures $\mu,\nu$ in $\R^3$ and $\cK$ is a highly nonlinear and nonlocal functional of $\mu$ and $\nu$ (involving another minimization problem). See Theorem~\ref{thm:main} below for the details. The measure $\mu$ is the one used to smear out the point charges, whereas the measure $\nu$ is used to estimate errors in the Onsager argument~\cite{Onsager-39,LieOxf-80,LieSei-09}.

In Section~\ref{sec:numerics}, we explain how we solved the minimization problem~\eqref{eq:general_form_proof} numerically, which gave a value slightly below $1.58$. This requires finding good trial states $\mu,\nu$ and to be able to compute $\cK(\mu,\nu)$ sufficiently accurately. Minimizers of the variational problem in~\eqref{eq:general_form_proof} may not be unique. Nothing guarantees that we have found a global minimizer but that does not affect the rigor of our $1.58$ bound. Other numerical techniques could possibly provide better trial states. We hope that our work will stimulate further activities on the Lieb-Oxford bound.

\subsubsection*{\textbf{The exchange energy}}
Our approach can also be used to prove a new estimate on the exchange energy.
Here we define the exchange energy the same as for $E_\text{ind}[\rho]$ in~\eqref{eq:E_ind} except that we require that $\bP$ arises from a Slater determinant $\Psi$, through the same formula as in~\eqref{eq:bP_Psi}. Such a probability measure $\bP$ is also called a determinantal point process~\cite{Soshnikov-00}. Written in terms of the one-particle density matrix $\gamma$ of the Slater determinant, this can be equivalently expressed as
\begin{equation}
E_\text{exc}[\rho]:=\inf_{\substack{\gamma=\gamma^2=\gamma^*\\ \rho_\gamma=\rho}}\left\{-\frac12\iint_{\R^3\times\R^3}\frac{|\gamma(x,y)|^2}{|x-y|}\dx\,\dy\right\}.
\label{eq:E_ex}
\end{equation}
We use the same notation $\gamma$ for the operator acting on $L^2(\R^3,\C^2)$ and for its integral kernel. The kernel $\gamma(x,y)$ is a $2\times2$ hermitian matrix and $|\gamma(x,y)|^2=\sum_{\sigma,\sigma'}|\gamma(x,y)_{\sigma,\sigma'}|^2$ denotes the Frobenius norm squared of this matrix. The total density is, by definition, given by
$$\rho_\gamma(x)=\sum_{\sigma\in\{\uparrow,\downarrow\}}\gamma(x,x)_{\sigma,\sigma}.$$
Any density $\rho\in L^1(\R^3)$ with $\int_{\R^3}\rho\in\N$ is representable by a Slater determinant~\cite[Thm.~1.2]{Lieb-83b}.

Since we have restricted the minimum to Slater determinants in~\eqref{eq:E_ex}, it is clear that $E_\text{exc}[\rho]\geq E_\text{ind}[\rho]$ and therefore the 1.58 bound~\eqref{eq:LO} holds as well. The authors of~\cite{PerRuzSunBur-14,PerSun-22} have raised the question of whether this is optimal, or if a better bound could be derived. In Theorem~\ref{thm:negative_correlations} we shall prove that
\begin{equation}
\boxed{E_\text{exc}[\rho]\geq -1.25\int_{\R^3}\rho(x)^{\frac43}\,\dx.}
\label{eq:LO_exchange}
\end{equation}
To our knowledge, this is the first bound for exchange that has a constant strictly lower than the best constant $c_{\rm LO}\geq 1.44$.

In fact we shall obtain the constant 1.25 for a more general class of states, not just Slater determinant. We will be able to handle all the classical probabilities $\bP$ which have a negative truncated two-particle correlation function, that is, satisfy the pointwise inequality
\begin{equation}
 \rho_\bP^{(2)}(x,y)- \rho_\bP(x)\rho_\bP(y)\leq0,\qquad\forall x,y\in\R^3,
 \label{eq:negative_correlations_intro}
\end{equation}
where
\begin{equation}
\rho^{(2)}_\bP(x,y):=N(N-1)\int_{(\R^3)^{N-2}}\rd\bP(x,y,x_3,...,x_N)
\label{eq:rho_2}
\end{equation}
denotes the two-point correlation function. This includes gas-like phases (at high temperature) and proves that the best Lieb-Oxford constant can never be attained with such gaseous phases. At constant density we will even lower the constant to $1.21$. These results are all described in Section~\ref{sec:exchange}.

From the exchange energy of the infinite non-interacting electron gas (as computed by Dirac in~\cite{Dirac-28b}), we know that the best constant in~\eqref{eq:LO_exchange} is at least $(3/4)\left(6/\pi\right)^{1/3}\simeq0.9305$. This is not optimal, however. One obtains a better lower bound using the one-particle problem $N=1$
\begin{equation}
\sup_{\substack{\rho\geq0\\ \int_{\R^3}\rho=1}}\frac{\iint_{\R^3\times\R^3}\frac{\rho(x)\rho(y)}{|x-y|}\dx\,\dy}{2\int_{\R^3}\rho(x)^{4/3}\,\dx}\simeq 1.0918.
\label{eq:N=1}
\end{equation}
The numerical value is obtained by solving the Lane-Emden equation~\cite{LieOxf-80}. It was conjectured in~\cite{PerRuzSunBur-14,PerSun-22} that the constant  in~\eqref{eq:N=1} might be optimal for the exchange inequality~\eqref{eq:LO_exchange}.

We emphasize that in~\eqref{eq:E_ex} we did not impose any spin symmetry. In fact, since our bound~\eqref{eq:LO_exchange} solely relies on the negative correlations~\eqref{eq:negative_correlations_intro} (hence is purely classical), it is completely \emph{independent on the number of spin states} (equal to two for electrons). In density functional theory, the exchange energy~\eqref{eq:E_ex} is sometimes rather defined by restricting the minimum to paramagnetic states, which take the special form $\gamma(x,y)_{\sigma\sigma'}=\tau(x,y)\delta_{\sigma\sigma'}$. Then $\rho_\gamma=2\rho_\tau$ and $|\gamma(x,y)|^2=2\tau(x,y)^2$. For such states, the Lieb-Oxford constant is therefore multiplied by $2^{-1/3}$ and we get $0.99$ in place of $1.25$.

\subsubsection*{\textbf{Acknowledgments.}} We would like to thank David Gontier for useful advice on the numerical simulations. This project has received funding from the European Research Council (ERC) under the European Union's Horizon 2020 research and innovation programme (grant agreements MDFT No 725528 of M.L. and AQUAMS No 694227 of R.S.). We are thankful for the hospitality of the Institut Henri Poincaré in Paris, where part of this work was done.

\subsubsection*{\textbf{Conflict of interest statement.}} The authors declare that they have no conflict of interest.

\subsubsection*{\textbf{Data availability statement.}} The manuscript has no associated data.

\section{A new estimate on the Lieb-Oxford constant}\label{sec:theory}
In this section we describe our new Lieb-Oxford bound. We denote by
$$D(\nu_1,\nu_2):=\frac12\iint_{\R^3\times\R^3}\frac{\rd \nu_1(x)\,\rd\nu_2(y)}{|x-y|}$$
the Coulomb scalar product of two finite measures $\nu_1,\nu_2$. We recall that Newton's theorem implies that for any radial probability measures $\nu_1$ and $\nu_2$, centered respectively at $R_1$ and $R_2$, we have
\begin{equation}
 2D(\nu_1,\nu_2)\leq 2D(\nu_1,\delta_{R_2})\leq 2D(\delta_{R_1},\delta_{R_2})=\frac{1}{|R_1-R_2|}.
 \label{eq:Newton}
\end{equation}
There is equality if $\nu_1$ and $\nu_2$ have disjoint supports.

Let $\rho\geq0$ be any non-negative function in $L^1(\R^3)\cap L^{4/3}(\R^3)$ with $\int_{\R^3}\rho=N\in\N$. The starting point is the exact same as in~\cite{LieOxf-80} and many other works on the subject~\cite{LieNar-75,LieSei-09,BenBleLos-12,LewLie-15}.
We consider a \emph{radial} probability measure $\mu$ (about the origin), such that $D(\mu,\mu)<\ii$. We then define as in~\cite{LieOxf-80}
\begin{equation}
 \mu_x(y):=\rho(x)\,\mu\left(\rho(x)^{\frac13}(y-x)\right).
 \label{eq:def_mu_x}
\end{equation}
In other words we recenter $\mu$ at the point $x$ and rescale it according to the local value of the given density $\rho$. By Newton's theorem~\eqref{eq:Newton}, we have
\begin{align*}
\sum_{1\leq j<k\leq N}\frac1{|x_j-x_k|}&\geq 2\sum_{1\leq j<k\leq N}D(\mu_{x_j},\mu_{x_k})\\
&=D\left(\sum_{j=1}^N\mu_{x_j},\sum_{j=1}^N\mu_{x_j}\right) -\sum_{j=1}^ND(\mu_{x_j},\mu_{x_j})\\
&=D\left(\sum_{j=1}^N\mu_{x_j},\sum_{j=1}^N\mu_{x_j}\right) -D(\mu,\mu)\sum_{j=1}^N\rho(x_j)^{\frac13}
\end{align*}
for all $x_1,...,x_N\in\R^3$. The last line is because $D(\mu_{x},\mu_{x})=\rho(x)^{\frac13} D(\mu,\mu)$ by scaling.
Following Onsager~\cite{Onsager-39}, we now use that for any measures $f,\eta$
\begin{align*}
D(f,f)&=D(f-\eta,f-\eta) -D(\eta,\eta)+2D(\eta,f)\\
&\geq -D(\eta,\eta)+2D(\eta,f).
\end{align*}
This is because the Coulomb potential has a positive Fourier transform, hence $D(f-\eta,f-\eta)\geq0$. Applying this to $f=\sum_{j=1}^N\mu_{x_j}$ we obtain
\begin{equation*}
\sum_{1\leq j<k\leq N}\frac1{|x_j-x_k|}\geq  -D(\eta,\eta)+2\sum_{j=1}^ND(\mu_{x_j},\eta)-D(\mu,\mu)\sum_{j=1}^N\rho(x_j)^{\frac13}.
\end{equation*}
The important point here is that we have estimated a two-body term by a one-body term, pointwise. Integrating against a probability $\bP$ of density $\rho$ and minimizing over $\bP$, we find the lower bound
\begin{equation}
E_{\rm ind}[\rho]\geq -D(\rho,\rho)  -D(\eta,\eta)+2\int_{\R^3}\rho(x)D(\mu_{x},\eta)\rd x-D(\mu,\mu)\int_{\R^3}\rho(x)^{\frac43}\,\rd x,
\label{eq:lower_eta}
\end{equation}
valid for all $\eta$. The last term is already of the desired form.
Lieb and Oxford chose $\eta=\rho$ and then the first three terms can be re-expressed as
\begin{equation}
 2\int_{\R^3}\rho(x)\,D\big(\mu_x-\delta_x,\rho)\,\rd x.
 \label{eq:error_LO}
\end{equation}
Estimating this term only in terms of $\int_{\R^3}\rho(x)^{4/3}\rd x$ is not easy and was achieved first in~\cite{Lieb-79,LieOxf-80}.

We take a different route and rather take $\eta$ of the form
$$\eta(y)=\int_{\R^3}\rho(x)\nu_x(y)\,\rd x$$
where $\nu$ is another radial probability measure and $\nu_x$ is defined exactly as for $\mu_x$ in~\eqref{eq:def_mu_x}. This yields the lower bound
\begin{multline}
E_{\rm ind}[\rho]\geq -\iint_{\R^3\times\R^3}\rho(x)\rho(y)\big(D(\nu_{x},\nu_y)+D(\delta_x,\delta_y)-2D(\mu_x,\nu_y)\big)\rd x\\
-D(\mu,\mu)\int_{\R^3}\rho(x)^{\frac43}\,\rd x.
\label{eq:new_lower}
\end{multline}
The Lieb-Oxford choice $\eta=\rho$ corresponds to $\nu=\delta_0$, whereas the best lower bound is obtained for $\nu=\mu$. Later we will make other approximations which will not preserve the order and, for this reason, we keep $\nu$ arbitrary. Our goal is to bound the first term on the right side of~\eqref{eq:new_lower}. It is more complicated to manipulate than the error term~\eqref{eq:error_LO} found by Lieb and Oxford.

First we rewrite the first term of~\eqref{eq:new_lower} in a different form. By scaling, $|x-y|D(\nu_x,\nu_y)$ and $|x-y|D(\mu_x,\nu_y)$ are in fact functions of $\rho(x)^{\frac13}|x-y|$ and $\rho(y)^{\frac13}|x-y|$ only. More explicitly, using the radial symmetry of $\mu$ and $\nu$, we can express
\begin{equation*}
D(\mu_x,\nu_y)=\frac1{2|x-y|}\iint_{\R^3\times\R^3}\frac{\rd\mu(u)\,\rd\nu(v)}{\left|e_1+\frac{u}{|x-y|\rho(x)^{\frac13}}-\frac{v}{|x-y|\rho(y)^{\frac13}}\right|},
\end{equation*}
where $e_1=(1,0,0)$. This suggests to introduce the function
\begin{equation}
\boxed{\Phi_{\mu\nu}(a,b):=a^3b^3\left(1-\iint_{\R^3\times\R^3}\frac{\rd\mu(u)\,\rd\nu(v)}{\left|e_1+u/a-v/b\right|}\right)}
\label{eq:def_Phi_mu}
\end{equation}
for $a,b>0$. A different way of writing the same is
\begin{equation}
\Phi_{\mu\nu}(a,b):=a^3b^3\Big(1-2D(\mu_{0,a},\nu_{e_1,b})\Big)
\label{eq:def_Psi_mu}
\end{equation}
where $\mu_{v,a}:=a^3\mu(a(\cdot-v))$ denotes the measure $\mu$ dilated by $1/a$ and placed at $v$. The term in the parenthesis of~\eqref{eq:def_Phi_mu} and~\eqref{eq:def_Psi_mu} represents the difference between the Coulomb interaction of two point charges placed at distance one, and the same system with the two point charges smeared using the radial measures $\mu$ and $\nu$, at the scales $1/a$ and $1/b$, respectively.
By~\eqref{eq:Newton}, we have
\begin{equation}
 2D(\mu_{0,a},\nu_{e_1,b})\leq 2D(\mu_{0,a},\delta_{e_1})=a\,V_\mu(ae_1)\leq1,
 \label{eq:estim_potential}
\end{equation}
where $V_\mu:=\mu\ast|x|^{-1}$ denotes the Coulomb potential generated by $\mu$. Hence $\Phi_{\mu\nu}$ is non-negative. It is continuous on $(0,\ii)^2$ and can be extended by continuity to $[0,\ii)^2$. It vanishes on the boundary due to the factor $a^3b^3$.
If the support of $\mu$ and $\nu$ are both included in the ball of radius $r$, it follows by Newton's theorem~\eqref{eq:Newton} that $\Phi_{\mu\nu}$ vanishes for $a^{-1}+b^{-1}\leq r$.

After symmetrizing in $x$ and $y$, we can express our error term~\eqref{eq:new_lower} as
\begin{multline}
\iint_{\R^3\times\R^3}\rho(x)\rho(y)\big(D(\nu_{x},\nu_y)+D(\delta_x,\delta_y)-2D(\mu_x,\nu_y)\big)\rd x\\
=\frac12\iint_{\R^3\times\R^3}\frac{\Psi_{\mu\nu}\left(|x-y|\rho(x)^{\frac13},|x-y|\rho(y)^{\frac13}\right)}{|x-y|^7}\rd x\,\rd y,
\label{eq:re-expressed_Psi_mu}
\end{multline}
with the new function
$$\boxed{\Psi_{\mu\nu}=\Phi_{\mu\nu}+\Phi_{\nu\mu}-\Phi_{\nu\nu}.}$$
The function $\Psi_{\mu\nu}$ is continuous on the quadrant $[0,\ii)^2$ and vanishes on the boundary. If $\nu=\mu$ then we simply have $\Psi_{\mu\mu}=\Phi_{\mu\mu}$ and then $\Psi_{\mu\mu}\geq0$. If $\nu=\delta_0$ as in~\cite{LieOxf-80}, then $\Psi_{\mu\delta_0}=\Phi_{\mu\delta_0}+\Phi_{\delta_0\mu}$ is also non-negative. In general, $\Psi_{\mu\nu}$ has no sign, however.

Our last task is to estimate the right side of~\eqref{eq:re-expressed_Psi_mu}. Our main new observation is that for any function $f$ so that
\begin{equation}
\Psi_{\mu\nu}(a,b)\leq f(a)+f(b),\qquad\forall a,b\geq0,
 \label{eq:Psi_fa_fb}
\end{equation}
we can immediately bound
\begin{multline*}
\frac12\iint_{\R^3\times\R^3}\frac{\Psi_{\mu\nu}\left(|x-y|\rho(x)^{\frac13},|x-y|\rho(y)^{\frac13}\right)}{|x-y|^7}\rd x\,\rd y\\
\leq \iint_{\R^3\times\R^3}\frac{f\left(|x-y|\rho(x)^{\frac13}\right)}{|x-y|^7}\rd x\,\rd y=\left(\int_{\R^3}\rho(x)^{\frac43}\rd x\right)\left(\int_{\R^3}\frac{f(|z|)}{|z|^7}\rd z\right).
\end{multline*}
The last equality is obtained by first integrating over $y$ using the new variable $z=(y-x)\rho(x)^{1/3}$. Inserting this in~\eqref{eq:new_lower} we obtain the Lieb-Oxford inequality
$$E_{\rm ind}[\rho]\geq -\left(\int_{\R^3}\frac{f(|z|)}{|z|^7}\rd z+D(\mu,\mu)\right)\int_{\R^3}\rho(x)^{\frac43}\,\rd x.$$
It remains to understand the class of $f$ so that~\eqref{eq:Psi_fa_fb} holds, and then to optimize over $\mu$, $\nu$ and $f$ to obtain the smallest possible constant.

We will assume throughout that $f$ is continuous on $\R_+$. Due to the strong divergence of $|z|^{-7}$ at the origin, the finiteness of the integral requires that $f(0)=0$. But then we have due to the constraint~\eqref{eq:Psi_fa_fb}
$$f(a)\geq \Psi_{\mu\nu}(a,0)-f(0)=0,\qquad \forall a\geq0$$
since $\Psi_{\mu\nu}$ vanishes on the boundary of $[0,\ii)^2$. Hence, even if $\Psi_{\mu\nu}$ has in general no sign, we have to consider non-negative functions $f$. For any radial probability measures $\mu,\nu$ so that $D(\mu,\mu)<\ii$, we thus introduce the set
\begin{equation}
 \cF_{\mu\nu}:=\Big\{f\in C^0(\R_+,\R_+)\ :\ \Psi_{\mu\nu}(a,b)\leq f(a)+f(b)\text{ for all $a,b\in\R_+$}\Big\}
 \label{eq:def_cF_mu}
\end{equation}
as well as the corresponding minimization problem
\begin{equation}
\boxed{ I(\mu,\nu):=\inf_{f\in\cF_{\mu\nu}}\int_{\R^3}\frac{f(|z|)}{|z|^7}\rd z.}
 \label{eq:def_cI_mu}
\end{equation}
In the definition of $\cF_{\mu\nu}$ we can freely replace $\Psi_{\mu\nu}$ by its positive part $(\Psi_{\mu\nu})_+$.
We will prove later in Lemma~\ref{lem:simple_estim_I_mu} that $\cF_{\mu\nu}$ is not empty and contains a function $f$ such that $\int_{\R^3}|z|^{-7}f(|z|)\rd z<\ii$, hence $I(\mu,\nu)<\ii$.

We note that if we replace $\mu$ and $\nu$ by $\mu_t=t^3\mu(t\cdot)$ and $\nu_t=t^3\nu(t\cdot)$, then $D(\mu_t,\mu_t)=tD(\mu,\mu)$ whereas
$\Psi_{\mu_t\nu_t}(a,b)=t^{-6}\Psi_{\mu\nu}(ta,tb)$ hence $\cF_{\mu_t\nu_t}=\{t^{-6}f(t\cdot)\ :\ f\in\cF_{\mu\nu}\}$ and $I(\mu_t,\nu_t)=t^{-2}I(\mu,\nu)$. After optimizing over $t$, we obtain our final scaling-invariant upper bound on the best Lieb-Oxford constant.

\begin{theorem}[Main estimate]\label{thm:main}
The best Lieb-Oxford constant
$$c_{\rm LO}=\sup_{\substack{\rho\geq0\\ \int_{\R^3}\rho\in\N}}\frac{-E_{\rm ind}[\rho]}{\int_{\R^3}\rho(x)^{\frac43}\rd x}$$
satisfies
\begin{equation}
\boxed{c_{\rm LO}\leq \frac{3}{2}\left(\inf_{\mu,\nu} 2I(\mu,\nu)D(\mu,\mu)^2\right)^{\frac13}}
\label{eq:main}
\end{equation}
where the infimum is over all radial probability measures $\mu,\nu$ such that $D(\mu,\mu)<\ii$.
\end{theorem}

As we explain below in Remark~\ref{rmk:link_LO}, our main bound~\eqref{eq:main} is strictly better than the Lieb-Oxford bound~\cite{LieOxf-80} which was later numerically optimized by Chan and Handy in~\cite{ChaHan-99}. This was confirmed by numerical simulations. We managed to construct two trial measures $\mu$ and $\nu$ (see Figure~\ref{fig:optimal_mu} below), such that the right side of~\eqref{eq:main} is slightly below $1.58$, after evaluating $I(\mu,\nu)$ and $D(\mu,\mu)$ numerically. This is how we get~\eqref{eq:LO}. The details of the numerical method are explained later in Section~\ref{sec:numerics}. In Section~\ref{sec:prop_Psi}, we discuss some useful mathematical properties of the variational problem $I(\mu,\nu)$, which also play a role in the numerical implementation.

\begin{remark}
\rm Theorem~\ref{thm:main} applies  the same to the grand-canonical indirect energy, for which $\int_{\R^3}\rho$ can be any positive real number~\cite{LewLieSei-19_ppt}.
\hfill $\diamond$
\end{remark}

\begin{remark}[Link with the Lieb-Oxford proof]\label{rmk:link_LO}\rm
Let us quickly explain the link with the Lieb-Oxford proof~\cite{LieOxf-80}  for $\nu=\delta_0$. In this case the function $\Psi_{\mu\delta_0}$ has separated variables and reads
$$\Psi_{\mu \delta_0}(a,b)=a^3b^3(1-bV_\mu(be_1))+b^3a^3(1-aV_\mu(ae_1)).$$
We recall that $V_\mu=\mu\ast|x|^{-1}$ denotes the Coulomb potential generated by~$\mu$.
Let $\zeta(a)$ be the \emph{smallest non-decreasing function above} $\chi(a)=a^3(1-aV_\mu(ae_1))$. If $\mu$ has compact support, then so does $\chi$ by Newton. Since $\chi$ is then bounded and behaves like $a^3$ at the origin, $\zeta$ is a bounded continuous function such that $\zeta(a)\sim a^3$ when $a\to0$. It turns out that the function $f(a)=a^3\zeta(a)$ belongs to $\cF_{\mu\delta_0}$, since
$$\Psi_{\mu \delta_0}(a,b)=a^3\chi(b)+b^3\chi(a)\leq a^3\zeta(b)+b^3\zeta(a)\leq a^3\zeta(a)+b^3\zeta(b).$$
This is because $\zeta$ is non-decreasing, hence
$$a^3\zeta(a)+b^3\zeta(b)-a^3\zeta(b)-b^3\zeta(a)=(a^3-b^3)(\zeta(a)-\zeta(b))\geq0.$$
Inserting this in~\eqref{eq:main} we find
\begin{equation}
c_{\rm LO}\leq \frac32\left(2 D(\mu,\mu)^2\int_{\R^3}\frac{\zeta(|z|)}{|z|^4}\rd z \right)^{\frac13}.
\label{eq:LO_bound_slight_improve}
\end{equation}
The integral converges since $\zeta(|z|)$ is bounded and behaves like $|z|^3$ at the origin. Lieb and Oxford did not consider the smallest increasing function $\zeta$ above $\chi$. Instead, they used the simpler (and generally larger) increasing function $\xi$ obtained by taking the primitive of the positive part of the derivative of $\chi$:
$$\zeta(a)\leq \xi(a)=\int_0^a \chi'(s)_+\rd s.$$
This way they obtained
$$\int_{\R^3}\frac{\zeta(|z|)}{|z|^4}\rd z\leq \int_{\R^3}\frac{\xi(|z|)}{|z|^4}\,\rd z=\int_{\R^3}\big(4V_\mu(x)+x\cdot\nabla V_\mu(x)\big)_+\rd x.$$
Plugging this into~\eqref{eq:LO_bound_slight_improve} provides the Lieb-Oxford estimate~\cite{LieOxf-80}
\begin{equation}
c_{\rm LO}\leq \frac32\left(2 D(\mu,\mu)^2\int_{\R^3}\big(4V_\mu(x)+x\cdot\nabla V_\mu(x)\big)_+\rd x \right)^{\frac13}.
\label{eq:LO_bound}
\end{equation}
For $\mu$ the uniform measure of the ball, Lieb and Oxford found $c_{\rm LO}\leq 1.68$ in~\eqref{eq:LO_bound}. After optimizing over $\mu$ numerically, Chan and Handy found $c_{\rm LO}\leq 1.64$ in \cite{ChaHan-99}.

One can show that the function $\zeta$ (hence also $\xi$) can never be an optimizer for $I(\mu,\delta_0)$. This implies that the constant obtained by taking $\nu=\delta_0$ in~\eqref{eq:main} is strictly below the one found in~\cite{LieOxf-80,ChaHan-99}. We optimized $I(\mu,\delta_0)$ numerically with respect to $\mu$ and only obtained the constant $1.63$ instead of $1.64$. It is really necessary to optimize $\nu$ as well in order to substantially decrease the constant.\hfill$\diamond$
\end{remark}

\section{An estimate on the exchange energy}\label{sec:exchange}
In this section we use the previous approach to provide a better Lieb-Oxford bound for a special class of states including Slater determinants (Hartree-Fock states). In particular we deduce a bound on the exchange energy.

The indirect energy of $\bP$ in~\eqref{eq:cE_ind_bP} can be expressed in terms of the two-point correlation function in~\eqref{eq:rho_2} as
$$\cE_{\rm ind}[\bP]=\frac12\iint_{\R^3\times\R^3}\frac{\rho^{(2)}_\bP(x,y)-\rho_\bP(x)\rho_\bP(y)}{|x-y|}\dx\,\dy.$$
The numerator involves the so-called \emph{truncated} two-point correlation function $\rho^{(2)}_\bP(x,y)-\rho_\bP(x)\rho_\bP(y)$. The following applies to any state which has negative truncated correlations.

\begin{theorem}[Lieb-Oxford bound for negatively--correlated states]\label{thm:negative_correlations}
Let $\bP$ be a symmetric probability on $(\R^3)^N$ which has a pointwise negative truncated two-point correlation function:
\begin{equation}
\rho^{(2)}_\bP(x,y)-\rho_\bP(x)\rho_\bP(y)\leq 0,\qquad \forall x,y\in\R^3.
\label{eq:negative_correlations}
\end{equation}
If $\rho_\bP\in L^{4/3}(\R^3)$, then we have
\begin{equation}
\boxed{\cE_{\rm ind}[\bP]\geq -1.2490\int_{\R^3}\rho_\bP(x)^{\frac43}\,\dx.}
\label{eq:LO_neg_correlation}
\end{equation}
If $\rho_\bP^{1/3}\in (L^1\cap L^\ii)(\R^3)$, we also have the inequality
\begin{equation}
\cE_{\rm ind}[\bP]
\geq
-\frac32 \left(\frac{\pi}{6}\right)^{\frac13}\|\rho_\bP\|_{L^\ii}^{\frac13}\left(\int_{\R^3}\rho_\bP(x)^{\frac13}\,\dx\right)^{\frac13}\left(\int_{\R^3}\rho_\bP(x)^{\frac43}\,\dx\right)^{\frac23}.
 \label{eq:simpler_exchange}
\end{equation}
When $\rho_\bP$ is constant on its support, this reduces to~\eqref{eq:LO_neg_correlation} with the constant $\frac32 \left(\frac{\pi}{6}\right)^{1/3}\simeq 1.2090$ in place of $1.2490$.
\end{theorem}

An important example of states satisfying~\eqref{eq:negative_correlations} is given by $\bP$ of the form~\eqref{eq:bP_Psi} where $\Psi=(N!)^{-\frac12}\det(\phi_j(x_k,\sigma_k))$ is a Slater determinant (Hartree-Fock state). In this case we have
$$\rho^{(2)}_\bP(x,y)-\rho_\bP(x)\rho_\bP(y)=-|\gamma(x,y)|^2\leq0$$
where $\gamma(x,y)_{\sigma\sigma'}=\sum_{j=1}^N\phi_j(x,\sigma)\overline{\phi_j(y,\sigma')}$ is the associated one-particle density matrix. We thus obtain the claimed estimate~\eqref{eq:LO_exchange} on the exchange energy. However, the class of states satisfying~\eqref{eq:negative_correlations} is more general.

The inequality~\eqref{eq:simpler_exchange} is non local and is only displayed for the convenience of the reader. It is more interesting for densities $\rho_\bP$ which are constant on their support, leading to the better Lieb-Oxford constant $1.2090$. Note that there exists an inequality similar to~\eqref{eq:simpler_exchange} for general states, with the constant $\frac35 \left(\frac{9\pi}{2}\right)^{1/3}\simeq 1.4508$~\cite{LieNar-75,LieOxf-80}. Hence $1.2090$ in~\eqref{eq:simpler_exchange} is an improvement over this constant. After taking a thermodynamic limit, this covers any translation-invariant point process with negative truncated two-point correlations. In statistical mechanics, this is typical of gas phases~\cite{Ruelle} at high temperature.

\begin{corollary}[Lieb-Oxford for negatively-correlated homogeneous processes]\label{cor:Jellium}
Let $\mathscr{P}$ be a translation-invariant point process on $\R^3$, with intensity $\rho>0$ and finite local second moment, such that
$$\rho^{(2)}_\mathscr{P}(x-y)\leq \rho^2,\qquad\forall x,y\in\R^3.$$
Then we have
\begin{multline}
\lim_{R\to\ii}\frac1{2|B_R|}\iint_{B_R\times B_R}\frac{\rho^{(2)}_\mathscr{P}(x-y)-\rho^2}{|x-y|}\dx\,\dy\\=\frac12\int_{\R^3}\frac{\rho^{(2)}_\mathscr{P}(\tau)-\rho^2}{|\tau|}\,\rd\tau\geq -\frac32 \left(\frac{\pi}{6}\right)^{\frac13}\rho^{\frac43}.
\label{eq:jellium_infinite}
\end{multline}
\end{corollary}

Due to the translation-invariance, we wrote $\rho^{(2)}_\mathscr{P}(x-y)$ instead of $\rho^{(2)}_\mathscr{P}(x,y)$. The left side of~\eqref{eq:jellium_infinite} is also called the Jellium energy per unit volume~\cite[Lem.~33]{Lewin-22} or the renormalized energy~\cite{BorSer-13,Leble-16,LebSer-17} of the point process~$\mathscr{P}$. We give the proof of Corollary~\ref{cor:Jellium} after the one of Theorem~\ref{thm:negative_correlations}. For the free electron gas we have
$$\rho^{(2)}_\mathscr{P}(\tau)-\rho^2=-\left|\int_{|k|^2\leq\mu}e^{ik\cdot\tau}\,\frac{\rd k}{(2\pi)^d}\right|^2$$
for the chemical potential $\mu=4\pi^2\rho^{2/d}|B_1|^{-2/d}$. In 3D we get the Dirac constant $(3/4)\left(6/\pi\right)^{1/3}\simeq0.9305$~\cite{Dirac-28b} after dividing by $2|\tau|\rho^{4/3}$ and integrating over $\tau$.

\begin{proof}[Proof of Theorem~\ref{thm:negative_correlations}]
Our proof is inspired by the one of Theorem~\ref{thm:main} described in the previous section, in the special case that  $\mu$ is the uniform measure of some sphere and $\nu=\delta_0$. No smearing of charges is needed, though, and we write the argument in a more direct way without explicitly introducing $\mu$ and $\nu$. Let  $\rho^{(2)}_T(x,y):=\rho^{(2)}_\bP(x,y)-\rho_\bP(x)\rho_\bP(y)$ denote the truncated two-point correlation function and let $\rho(x):=\rho_\bP(x)$ be the density. We write
\begin{align*}
&\frac12\iint_{\R^3\times\R^3}\frac{\rho^{(2)}_T(x,y)}{|x-y|}\dx\,\dy\\
&=\frac12\iint_{\R^3\times\R^3}\rho^{(2)}_T(x,y)\left(\frac1{|x-y|}-\lambda\rho(x)^{\frac13}-\lambda\rho(y)^{\frac13}\right)\dx\,\dy-\lambda\int_{\R^3}\rho^{\frac43}\\
&=\frac12\iint_{\R^3\times\R^3}\frac{\rho^{(2)}_T(x,y)}{\rho(x)\rho(y)}\frac{\Theta\Big(\lambda\rho(x)^{\frac13}|x-y|,\lambda\rho(y)^{\frac13}|x-y|\Big)}{\lambda^6|x-y|^7}\dx\,\dy-\lambda\int_{\R^3}\rho^{\frac43}
\end{align*}
where, this time,
\begin{equation}
\boxed{\Theta(a,b)=a^3b^3(1-a-b).}
\label{eq:def_Theta}
\end{equation}
In the second line we have used the sum rule
\begin{equation}
 \int_{\R^3}\rho_T^{(2)}(x,y)\,\dy=\int_{\R^3}\rho_T^{(2)}(y,x)\,\dy=-\rho(x),
 \label{eq:sum_rule}
\end{equation}
since $\rho^{(2)}_T$ comes from an $N$-particle probability $\bP$.
If $\Theta(a,b)\leq f(a)+f(b)$ with $f\geq0$, we obtain
\begin{align*}
\cE_{\rm ind}[\bP]&=\frac12\iint_{\R^3\times\R^3}\frac{\rho^{(2)}_T(x,y)}{|x-y|}\dx\,\dy\\
&\geq \iint_{\R^3\times\R^3}\frac{\rho^{(2)}_T(x,y)}{\rho(x)\rho(y)}\frac{f\Big(\lambda\rho(x)^{\frac13}|x-y|\Big)}{\lambda^6|x-y|^7}\dx\,\dy-\lambda\int_{\R^3}\rho^{\frac43}\\
&\geq -\iint_{\R^3\times\R^3}\frac{f\Big(\lambda\rho(x)^{\frac13}|x-y|\Big)}{\lambda^6|x-y|^7}\dx\,\dy-\lambda\int_{\R^3}\rho^{\frac43}\\
&=-\left(\lambda^{-2}\int_{\R^3}\frac{f(|x|)}{|x|^7}\dx+\lambda\right)\int_{\R^3}\rho^{\frac43},
\end{align*}
where we have used that $0\geq \rho^{(2)}_T\geq -\rho(x)\rho(y)$.
Similarly to $I(\mu,\nu)$ in~\eqref{eq:def_cI_mu}, we introduce
\begin{equation}
J:=\inf_{\substack{f\in C^0([0,1],\R)\\ \Theta(a,b)\leq f(a)+f(b)}}\int_{|x|\leq 1}\frac{f(|x|)}{|x|^7}\dx
 \label{eq:def_J}
\end{equation}
and obtain after optimizing over $\lambda$
$$\boxed{\cE_{\rm ind}[\bP]\geq-\frac32\left(2J\right)^{\frac13}\int_{\R^3}\rho(x)^{\frac43}\,\dx.}$$
This is our improved Lieb-Oxford bound for states with negative truncated correlations.

In order to provide an estimate on $J$, we introduce the following function
\begin{align}
 g(a)&:=\max_{0\leq b\leq a}\left\{\Theta(a,b)-\frac{\Theta(b,b)_+}2\right\}\nn\\
 &=\max_{0\leq b\leq a}\left\{a^3b^3(1-a-b)-\frac{b^6}2(1-2b)_+\right\}.
 \label{eq:def_g_neg_corr}
\end{align}
A similar function $g$ plays a role for $I(\mu,\nu)$ later in Lemma~\ref{lem:G}. Taking $b=0$ we obtain $g\geq0$ and taking $b=a$ we find that $g(a)\geq \Theta(a,a)_+/2$. We also see that $g$ is supported on $[0,1]$ with $g(1)=0$.
We claim that $g$ is admissible for the infimum in~\eqref{eq:def_J}. This is because by definition we have
$$\Theta(a,b)\leq g(a)+\frac{\Theta(b,b)_+}2\leq g(a)+g(b),\qquad\forall 0\leq b\leq a.$$
By symmetry of $\Theta$, this proves that $g$ satisfies the constraint for all $a,b$, and thus
\begin{equation}
 J\leq \int_{\R^3}\frac{g(|x|)}{|x|^7}\,\dx=4\pi\int_0^1 a^{-5}g(a)\,\rd a.
 \label{eq:estim_J}
\end{equation}
We have observed numerically that $g$ is in fact the \emph{exact minimizer} of~$J$. More about this can be read in Remark~\ref{rmk:optimality_G} below. However we do not need to prove this optimality for the upper bound.

The integral of $g$ on the right side of~\eqref{eq:estim_J} can be computed to an arbitrary precision. We have
$$\frac\partial{\partial b}\left(\Theta(a,b)-\frac{\Theta(b,b)}{2}\right)=b^2(b-a)\Big(-3 a^2(1-a) + (7 a-3)a b + (7 a-3) b^2 + 7 b^3\Big)$$
where the third degree polynomial in the parenthesis has only one real root $b=R(a)$ for $a\in[0,1]$ (the other two are complex). It is possible to compute $R(a)$ exactly but we do not display its expression here. The maximum in~\eqref{eq:def_g_neg_corr} is attained at
\begin{equation}
b(a)=\min\big(a,R(a)\big)=\begin{cases}
a&\text{for $0\leq a\leq3/8,$}\\
R(a)&\text{for $3/8\leq a\leq1,$}
\end{cases}
 \label{eq:b_a}
\end{equation}
(see Figure~\ref{fig:roots_g} below). In fact, for $a\leq 3/8$, $b=a$ is the only possible maximum on $[0,a]$ and for $a>3/8$, $b=a$ becomes a strict local minimum. The maximum must thus be attained at $b=R(a)$. We conclude that
$$g(a)=\Theta\big(a,b(a)\big)-\frac{\Theta\big(b(a),b(a)\big)}2.$$
One can then evaluate the integral in~\eqref{eq:estim_J} to an arbitrary precision, leading to the bound $J\leq 0.2887$ and to the LO constant $1.2490$ claimed in~\eqref{eq:LO_neg_correlation}.

In order to prove the non-local bound~\eqref{eq:simpler_exchange}, we argue similarly but do not symmetrize in $x$ and $y$:
\begin{align*}
&\frac12\iint_{\R^3\times\R^3}\frac{\rho^{(2)}_T(x,y)}{|x-y|}\dx\,\dy\nn\\
&\qquad=\frac12\iint_{\R^3\times\R^3}\rho^{(2)}_T(x,y)\left(\frac1{|x-y|}-2\lambda\rho(x)^{\frac13}\right)\dx\,\dy-\lambda\int_{\R^3}\rho^{\frac43}\nn\\
&\qquad\geq\frac12\iint_{\R^3\times\R^3}\rho^{(2)}_T(x,y)\left(\frac1{|x-y|}-2\lambda\rho(x)^{\frac13}\right)_+\dx\,\dy-\lambda\int_{\R^3}\rho^{\frac43}\nn\\
&\qquad\geq-\frac12\iint_{\R^3\times\R^3}\rho(x)\rho(y)\left(\frac1{|x-y|}-2\lambda\rho(x)^{\frac13}\right)_+\dx\,\dy-\lambda\int_{\R^3}\rho^{\frac43}.
\end{align*}
Now we estimate $\rho(y)\leq\|\rho\|_{L^\ii}$ and integrate over $y$ to obtain
\begin{equation}
 \cE_{\rm ind}[\bP]=\frac12\iint_{\R^3\times\R^3}\frac{\rho^{(2)}_T(x,y)}{|x-y|}\dx\,\dy\geq-\frac{\pi\|\rho\|_{L^\ii}}{12\lambda^2}\int_{\R^3}\rho^{\frac13}-\lambda\int_{\R^3}\rho^{\frac43}.
 \label{eq:simple_rho}
\end{equation}
Optimizing over $\lambda$ yields~\eqref{eq:simpler_exchange}.
\end{proof}

\begin{proof}[Proof of Corollary~\ref{cor:Jellium}]
Denoting the number of particles in a bounded domain $A$ by $\cN_A$, we have for a general translation-invariant point process with finite local second moment
$$ \iint_{A\times A}\rho^{(2)}_T(x-y)\,\dx\,\dy=\bE[\cN_A(\cN_A-1)]-\bE[\cN_A]^2\geq -\bE[\cN_A]=-\rho|A|.$$
Taking $A=B_R$ and passing to the limit $R\to\ii$ gives
\begin{equation}
 \int_{\R^3}\rho^{(2)}_T(\tau)\,\rd\tau\geq-\rho.
 \label{eq:replaces_sum_rule}
\end{equation}
We then apply the previous estimate~\eqref{eq:simple_rho} to the restriction $\mathscr{P}_R$ of $\mathscr{P}$ to the ball $B_R$, using~\eqref{eq:replaces_sum_rule} in place of the sum rule~\eqref{eq:sum_rule}.
\end{proof}

\begin{figure}[t]
\includegraphics[width=6.2cm]{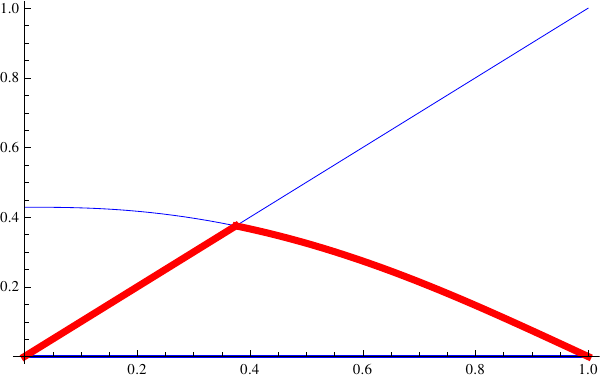} \includegraphics[width=6.2cm]{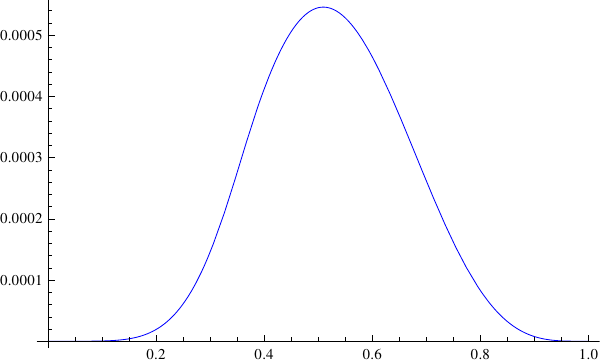}
\caption{\footnotesize \emph{Left:} Critical points of the function $b\mapsto \Theta(a,b)-\Theta(b,b)/2$ in terms of $a\in[0,1]$. The maximum in~\eqref{eq:def_g_neg_corr} is attained at $b=b(a)$, which is the minimum  of the two positive roots (thicker curve on the picture). It is equal to $a$ for $a\leq3/8$ and is then decreasing and reaches 0 at $a=1$. \emph{Right:} The corresponding function $g(a)=\Theta(a,b(a))-\Theta(b(a),b(a))/2$, which equals $\Theta(a,a)/2=a^6(1-2a)$ for $a\leq3/8$.\label{fig:roots_g}}
\end{figure}

\begin{remark}[Optimality of $g$]\label{rmk:optimality_G}\rm
We have $b(a)\leq 3/8<1/2$ for all $a\in[0,1]$ since the second root $R(a)$ is decreasing. This implies that $g(b(a))=\Theta(b(a),b(a))_+/2$ for all $a$ and, therefore, we can rewrite the definition~\eqref{eq:def_g_neg_corr} in the form of a fixed point equation
\begin{equation}
g(a)=\max_{0\leq b\leq a}\big\{\Theta(a,b)-g(b)\big\}.
\label{eq:g_SCF}
\end{equation}
As we have seen in the proof we also know that $g(a)\geq\Theta(a,b)-g(b)$ for all $a,b\in[0,1]$, which implies that
$$g(a)\geq \max_{0\leq b\leq 1}\big\{\Theta(a,b)-g(b)\big\}\geq \max_{0\leq b\leq a}\big\{\Theta(a,b)-g(b)\big\}=g(a).$$
Hence there must be equality and we have as well
\begin{equation}
g(a)=\max_{0\leq b\leq 1}\big\{\Theta(a,b)-g(b)\big\},
\label{eq:g_SCF2}
\end{equation}
where $b$ is now allowed to reach $1$ in the maximum. In fact, any solution to~\eqref{eq:g_SCF} must clearly be a solution to~\eqref{eq:g_SCF2}.
As we will prove later in Theorem~\ref{thm:existence_f} in a slightly more general context, the minimizing problem $J$ in~\eqref{eq:def_J} possesses minimizers, which must all solve the fixed point equation~\eqref{eq:g_SCF2}. This makes $g$ a good candidate for being a minimizer for $J$.

It turns out that $g$ is the \emph{unique non-negative solution to the first equation~\eqref{eq:g_SCF}}. To prove this fact, consider another solution $f\geq0$. Taking $b=a$ we find $f(a)\geq\Theta(a,a)_+/2$ and thus
$$f(a)=\max_{0\leq b\leq a}\big\{\Theta(a,b)-f(b)\big\}\leq \max_{0\leq b\leq a}\big\{\Theta(a,b)-\Theta(b,b)_+/2\big\}=g(a).$$
But then we also have
$$f(a)=\max_{0\leq b\leq a}\big\{\Theta(a,b)-f(b)\big\}\geq \max_{0\leq b\leq a}\big\{\Theta(a,b)-g(b)\big\}=g(a)$$
which shows that $f=g$.

On the other hand, $g$ is \emph{not} the unique solution of the second equation~\eqref{eq:g_SCF2}. There are many other solutions. One example is obtained by requiring $b\geq a$ instead of $b\leq a$:
$$h(a)=\max_{a\leq b\leq 1}\left\{\Theta(a,b)-\frac{\Theta(b,b)_+}2\right\}.$$
Similar arguments as for $g$ show that $h$ solves the equation~\eqref{eq:g_SCF2}. However, $h$ behaves like $a^3$ at the origin, so that $\int |x|^{-7}h(|x|)\,\dx=+\ii$ and $h$ is not an interesting solution for our minimization problem $J$ in~\eqref{eq:def_J}. While there are many solutions to the equation~\eqref{eq:g_SCF2}, $g$ is clearly the smallest on $[0,3/8]$ (where it coincides with the lower bound $\Theta(a,a)_+/2$) and hence we expect it to be the optimizer for $J$.
\hfill$\diamond$
\end{remark}

\section{Some properties of $\Psi_{\mu\nu}$ and $I(\mu,\nu)$}\label{sec:prop_Psi}

In this section we prove some properties of $\Psi_{\mu\nu}$ and $I(\mu,\nu)$ appearing in our Lieb-Oxford bound~\eqref{eq:main}, which will be useful in the numerical implementation described in Section~\ref{sec:numerics}.

First, we can compute $\Psi_{\mu\nu}$ explicitly for $\mu=\nu=\sigma$, the uniform measure of a sphere.

\begin{lemma}[Case of the sphere]\label{lem:sphere}
Let $\sigma$ be the uniform measure of the unit sphere, normalized as $\sigma(\R^3)=\sigma(\bS^2)=1$. Then we have
\begin{equation}
\Psi_{\sigma\sigma}(a,b)=\Phi_{\sigma\sigma}(a,b)=\frac{a^2b^2}{4}\Big((a+b-ab)_+^2-(|a-b|-ab)_+^2\Big).
\label{eq:Psi_sphere}
\end{equation}
In particular, $\Psi_{\sigma\sigma}$ is $C^1$ on $(\R_+) ^2$ with the bound
\begin{equation}
 |\nabla\Psi_{\sigma\sigma}(a,b)|\leq C(a^2b^3+a^3b^2).
 \label{eq:Psi_sphere_derivative}
\end{equation}
\end{lemma}

The proof is a tedious but elementary computation which we do not detail here.

Let us now turn our attention to general radial Borel measures $\mu,\nu$. Recall that later we will have to impose $D(\mu,\mu)<\ii$, which implies that $\mu(\{0\})=0$, that is, $\mu$ cannot include a delta measure at the origin. The radial measure $\nu$ is itself arbitrary and we write $\nu=\nu_0\delta_0+\nu'$, with $\nu'(\R^3)=1-\nu_0$ and $\nu'(\{0\})=0$.
We can then write the measures $\mu$ and $\nu'$ as a convex combination of uniform measures over spheres, which provides the formula
\begin{align}
\Psi_{\mu\nu}(a,b)=&\int_0^\ii\int_0^\ii\Psi_{\sigma\sigma}\left(\frac{a}{r},\frac{b}{s}\right) r^3s^3\rd m(r)\rd n(s)\nn\\
&+\int_0^\ii\int_0^\ii\Psi_{\sigma\sigma}\left(\frac{a}{r},\frac{b}{s}\right) r^3s^3\rd n(r)\rd m(s)\nn\\
&-\int_0^\ii\int_0^\ii\Psi_{\sigma\sigma}\left(\frac{a}{r},\frac{b}{s}\right) r^3s^3\rd n(r)\rd n(s)+2\nu_0^2a^3b^3\nn\\
&-a^3b^3\nu_0 \big(aV_{\nu'}(ae_1)+aV_\mu(ae_1)\big)-a^3b^3\nu_0 \big(bV_{\nu'}(be_1)+bV_\mu(be_1)\big)
\label{eq:Psi_mu_average_spheres}
\end{align}
for some Borel probability measure $m,n$ on $\R$ so that $m\big((0,\ii)\big)=1$ and $n\big((0,\ii)\big)=1-\nu_0$.
By Newton's theorem, we have
\begin{equation}
aV_\mu(ae_1)=\mu(\{|x|\leq a\})+a\int_{|x|\geq a}\frac{\rd \mu(x)}{|x|}
\label{eq:limit_V_mu}
\end{equation}
and a similar formula for $\nu'$. This implies that $aV_\mu(ae_1)$ has a locally bounded derivative on $(0,\ii)$, with
$$\frac{\rd}{\rd a} aV_\mu(ae_1)=\int_{|x|\geq a}\frac{\rd \mu(x)}{|x|}\leq \frac{1}{a}.$$
Thanks to the multiplying factor $a^3b^3$, the terms on the last line of~\eqref{eq:Psi_mu_average_spheres} are thus $C^{0,1}$ on $[0,\ii)$, with a derivative bounded by $C(a^2b^3+a^3b^2)$. The theory of integrals depending on a parameter and the estimate~\eqref{eq:Psi_sphere_derivative} then imply the following.

\begin{lemma}\label{lem:regularity_Psi}
Let $\mu$ and $\nu$ be radial probability measures such that
\begin{equation}
\mu(\{0\})=0,\qquad \int |x|\,\rd \mu(x)<\ii,\qquad \int |x|\,\rd \nu(x)<\ii.
 \label{eq:assumption_mu_nu}
\end{equation}
Then the function $\Psi_{\mu\nu}$ has a locally bounded derivative on $[0,\ii)^2$, with the same estimate~\eqref{eq:Psi_sphere_derivative}. It vanishes on the two axis $\{a=0\}\cup\{b=0\}$.
If $\nu(\{0\})=0$, then $\Psi_{\mu\nu}$ is in fact $C^1$.
\end{lemma}

It will be convenient to assume that $\mu$ and $\nu$  have compact support.
Due to the scaling invariance of~\eqref{eq:main}, we can then always suppose that the support is included in the unit ball $B_1$. The following is a simple preliminary bound.

\begin{lemma}[Well-posedness of $I(\mu,\nu)$]\label{lem:simple_estim_I_mu}
For any radial probability measures $\mu$ and $\nu$ supported in the unit ball, we have for the function $\Phi_{\mu\nu}$ in~\eqref{eq:def_Phi_mu}
\begin{equation}
0\leq \Phi_{\mu\nu}(a,b)\leq a^3b^3\1(a^{-1}+b^{-1}>1)\leq \min(a^6, 4a^3)+\min(b^6, 4b^3)
 \label{eq:simple_F}
\end{equation}
and hence
\begin{equation}
\Psi_{\mu\nu}(a,b)\leq 2a^3b^3\1(a^{-1}+b^{-1}>1)\leq \min(2a^6, 8a^3)+\min(2b^6, 8b^3).
 \label{eq:simple_F_Psi}
\end{equation}
This implies that the function $f(a)=\min(2a^6, 8a^3)$ always belongs to $\cF_{\mu\nu}$ and gives
\begin{equation}
I(\mu,\nu)\leq 2^{\frac43}12\pi,\qquad I(\mu,\mu)\leq 2^{\frac13}12\pi
\label{eq:simple_estim_I_mu}
\end{equation}
for all such $\mu$ and $\nu$.
\end{lemma}

\begin{proof}
The definition~\eqref{eq:def_Phi_mu} gives immediately $\Phi_{\mu\nu}(a,b)\leq a^3b^3$ whereas
Newton's theorem implies that $\Phi_{\mu\nu}$ vanishes on the set $\{a^{-1}+b^{-1}\leq1\}$. This gives the first inequality in~\eqref{eq:simple_F}. We now prove the second. If $a,b\leq 2^{\frac23}$, then we write
$a^3b^3\leq a^6+b^6$.
If $a> 2^{\frac23}$ and $b\leq 2^{\frac23}$, we simply use  $a^3b^3\leq 4a^3\leq 4a^3+b^6$.
Finally, if $a,b> 2^{\frac23}$ we only have to consider the case $b<a/(a-1)$, due to the characteristic function. We use that $a^3\leq 4+4(a-1)^3$ for $a>2^{\frac23}$, which gives
$$b^3\leq \frac{a^3}{(a-1)^3}\leq \frac{4a^3}{a^3-4},$$
and thus
$ a^3b^3\leq 4(a^3+b^3)$,
as desired. If $\mu=\nu$, then we have $\Psi_{\mu\mu}=\Phi_{\mu\mu}$ and thus conclude that $\phi(a)=\min(a^6,4a^3)$ is in $\cF_{\mu\mu}$. The second estimate~\eqref{eq:simple_estim_I_mu} follows after plugging $\phi$ in~\eqref{eq:def_cI_mu}. If $\nu\neq\mu$, then we use $\Psi_{\mu\nu}\leq \Phi_{\mu\nu}+\Phi_{\nu\mu}$, which proves~\eqref{eq:simple_F_Psi} and that $f=2\phi$ is in $\cF_{\mu\nu}$. We get the extra factor $2$ in the first inequality of~\eqref{eq:simple_estim_I_mu}.
\end{proof}

Taking $\mu=\nu=\sigma$ and plugging~\eqref{eq:simple_estim_I_mu} into~\eqref{eq:main}, we obtain the simple bound $c_{\rm LO}\leq 4.3117$. This is not a very good bound but its proof is completely elementary.

The following is another function in $\cF_{\mu\nu}$ which will be useful, though less explicit.

\begin{lemma}\label{lem:G}
Let $\mu,\nu$ be two radial probability measures supported in the unit ball, such that $\mu(\{0\})=0$.
Define
\begin{equation}
\boxed{g(a):=\max_{0\leq b\leq a}\left\{\Psi_{\mu\nu}(a,b)-\frac{\Psi_{\mu\nu}(b,b)_+}{2}\right\}. }
\label{eq:def_G}
\end{equation}
Then we have $g\in\cF_{\mu\nu}$ with
\begin{equation}
 c\min(a^6,a^3)\leq g(a)\leq  2a^6\left(\1(a\leq 2)+\frac{\1(a>2)}{(a-1)^3}\right)\leq \min(2a^6,16a^3)
 \label{eq:simple_upper_G}
\end{equation}
for some $c>0$ depending only on $\mu$ and $\nu$. At infinity, we have
\begin{equation}
g(a)\underset{a\to\ii}{\sim} a^3\max_{0\leq b\leq1}b^3\big(1-bV_\mu(be_1)\big).
\label{eq:G_infinity}
\end{equation}
\end{lemma}

\begin{proof}
Note that $g\geq0$ since at $b=0$ we find $\Psi_{\mu\nu}(a,0)=\Psi_{\mu\nu}(0,0)=0$ by Lemma~\ref{lem:regularity_Psi}. By taking $b=a$ we obtain $g(a)\geq\Psi_{\mu\nu}(a,a)_+/2$.
If $b\leq a$ we have by definition
$$\Psi_{\mu\nu}(a,b)\leq g(a)+\frac{\Psi_{\mu\nu}(b,b)_+}{2}\leq g(a)+g(b).$$
By symmetry we conclude that $g\in\cF_{\mu\nu}$.
Since $\Psi_{\mu\nu}(a,b)\leq 2a^3b^3$, we have $g(a)\leq 2a^6$. To improve the bound for large $a$, we recall that $\Psi_{\mu\nu}(a,b)=0$ for $a^{-1}+b^{-1}\leq 1$. Therefore, whenever $a>1$, we have
$$g(a)\leq \max_{0\leq b\leq \frac{a}{a-1}}\Psi_{\mu\nu}(a,b)\leq 2\frac{a^6}{(a-1)^3}.$$
This concludes the proof of the upper bound in~\eqref{eq:simple_upper_G}.

Next we derive the lower bound on $g$. We assume first that $\nu\neq \delta_0$. Using Newton's theorem we obtain
\begin{align}
\Psi_{\mu\nu}(a,b)&=a^3b^3\big(1+2D(\nu_{0,a},\nu_{e_1,b})-2D(\nu_{0,a},\mu_{e_1,b})-2D(\mu_{0,a},\nu_{e_1,b})\big) \nn\\
&\geq a^3b^3\big(1-bV_\mu(be_1)-bV_\nu(be_1)\big)\label{eq:lower_bd_Psi}
\end{align}
and
$$\Psi_{\mu\nu}(a,b)\leq a^3b^3\big(1+bV_\nu(be_1)\big).$$
We can obtain a bound for all $a$ using~\eqref{eq:lower_bd_Psi} as follows
\begin{equation}
g(a)\geq \frac{\Psi_{\mu \nu}(a,a)_+}2\geq \frac{a^6}2\big(1-aV_\mu(ae_1)-aV_\nu(ae_1)\big)_+.
\label{eq:lower_bd_G}
\end{equation}
For small $a$, we have due to~\eqref{eq:limit_V_mu}
$$aV_\mu(ae_1)=\mu(\{|x|\leq a\})+a\int_{|x|\geq a}\frac{\rd \mu(x)}{|x|}\underset{a\to0}{\longrightarrow}\mu(\{0\})=0,$$
and similarly $aV_\nu(ae_1)\to \nu(\{0\})$. Thus the function on the right of~\eqref{eq:lower_bd_G} behaves like $a^6(1-\nu(\{0\}))/2>0$ for $a$ small. This proves that $g(a)\geq a^6(1-\nu(\{0\}))/4$ for $a\leq a_0$ small enough. For larger $a$'s we write instead
\begin{align*}
g(a)&\geq a^3b^3\big(1-bV_\mu(be_1)-bV_\nu(be_1)\big)-\frac{b^6}2\big(1+bV_\nu(be_1)\big)\\
&\geq \frac{a^3b^3}2\big(1-bV_\mu(be_1)-bV_\nu(be_1)\big)
\end{align*}
for a small but fixed $b\leq a_0$ chosen such that
$$a_0^3\big(1-bV_\mu(be_1)-bV_\nu(be_1)\big)\geq b^3\big(1+bV_\nu(be_1)\big).$$
This proves the lower bound $g(a)\geq ca^3$ for $a\geq a_0$.

When $\nu=\delta_0$ we can use that
\begin{equation}
\Psi_{\mu\delta_0}(a,b)=a^3b^3(2-aV_\mu(ae_1)-bV_\mu(be_1))\geq a^3b^3(1-bV_\mu(be_1))
\label{eq:formula_Psi_delta}
\end{equation}
and proceed similarly. This concludes the proof of~\eqref{eq:simple_upper_G}.

For the behavior~\eqref{eq:G_infinity} at infinity, we choose $b\in(0,1)$ realizing the maximum on the right side and use that $2D(\nu_{0,a},\nu_{e_1,b})\to bV_\nu(be_1)$ when $a\to\ii$, and a similar convergence for the two other terms of the first line of~\eqref{eq:lower_bd_Psi}. This proves that
$$\liminf_{a\to\ii}\frac{g(a)}{a^3}\geq \max_{0\leq b\leq1} b^3\big(1-bV_\mu(be_1)\big).$$
To show the reverse inequality, we take $a_n\to\ii$ realizing the limsup and call $b_n$ an optimizer for the maximum defining $g(a_n)$. We must have $b_n\leq a_n/(a_n-1)$, hence $b_n$ is bounded. After extraction of a subsequence, we can assume $b_n\to b$. Then
$$\frac{g(a_n)}{a_n^3}\leq b_n^3\big(1+b_nV_\nu(b_ne_1)-2D(\nu_{0,a_n},\mu_{e_1,b_n})-2D(\mu_{0,a_n},\nu_{e_1,b_n})\big)$$
and we use that the right side converges to $b^3(1-bV_\mu(be_1))$ when $a_n\to\ii$ and $b_n\to b$.
\end{proof}

We now introduce a truncated problem, which we will simulate on the computer. For $R>0$ we define
\begin{equation}
 \cF_{\mu\nu,R}:=\Big\{f\in C^0([0,R],\R_+)\ :\ \Psi_{\mu\nu}(a,b)\leq f(a)+f(b)\text{ for all $0\leq a,b\leq R$}\Big\}
 \label{eq:def_cF_mu_R}
\end{equation}
as well as the corresponding minimization problem
\begin{equation}
\boxed{ I_R(\mu,\nu):=\inf_{f\in\cF_{\mu\nu,R}}\int_{|z|\leq R}\frac{f(|z|)}{|z|^7}\rd z,}
 \label{eq:def_cI_mu_R}
\end{equation}
that is, we work on $[0,R]$ instead of $\R_+$. Since the restriction to $[0,R]$ of a function in $\cF_{\mu\nu}$ always belongs to $\cF_{\mu\nu,R}$, we obviously have the inequality
$$I(\mu,\nu)\geq I_R(\mu,\nu).$$
We are thus approaching $I(\mu,\nu)$ from below.
The following provides a quantitative estimate between the truncated and the original problems.

\begin{lemma}[Speed of convergence]\label{lem:speed_CV}
Let $\mu,\nu$ be two radial probability measures supported in the unit ball, such that $\mu(\{0\})=0$. We have
\begin{equation}
 I_R(\mu,\nu)\leq I(\mu,\nu)\leq I_R(\mu,\nu)+\int_{\R^3\setminus B_R}\frac{g(|z|)}{|z|^7}\rd z
 \label{eq:estimate_I_R}
\end{equation}
where $g$ is the function in Lemma~\ref{lem:G}.
\end{lemma}

Using~\eqref{eq:simple_upper_G} to estimate the integral of $g$, we obtain the explicit bound
\begin{equation}
I_R(\mu,\nu)\leq I(\mu,\nu) \leq I_R(\mu,\nu)+\frac{8\pi R^2}{(R-1)^3}.
 \label{eq:explicit_upper_I_R}
\end{equation}
In particular, $I_R(\mu,\nu)=I(\mu,\nu)+O(R^{-1})$. However~\eqref{eq:explicit_upper_I_R} is not a very good bound since at infinity $g(a)\sim \kappa a^3$ with $\kappa<1$ by~\eqref{eq:G_infinity}.

\begin{proof}
Let $f$ be any function of $\cF_{\mu\nu,R}$ and define the extension
$$\tilde f=f\1_{[0,R]}+g\1_{(R,\ii)}.$$
We claim that
$$\Psi_{\mu\nu}(a,b)\leq \tilde f(a)+\tilde f(b),\qquad \forall a,b\geq0.$$
Since $f$ satisfies this property on $[0,R]$ already and $g\in\cF_{\mu\nu}$, we only have to look at the case where, say, $a\leq R<b$. By definition of $g$ we then have
$$\Psi_{\mu\nu}(a,b)\leq \frac{\Psi_{\mu\nu}(a,a)_+}2+g(b)\leq f(a)+g(b)=\tilde f(a)+\tilde f(b).$$
Using $\tilde f$ as a trial function for $I(\mu,\nu)$ we get
\begin{equation}
 I_R(\mu,\nu)\leq I(\mu,\nu)\leq \int_{B_R}\frac{f(|z|)}{|z|^7}\,\rd z+\int_{\R^3\setminus B_R}\frac{g(|z|)}{|z|^7}\rd z.
  \label{eq:estimate_I_R_f}
\end{equation}
To be more precise, $\tilde f$ is not necessarily a continuous function but it has at most one jump, at $R$. It can thus be approximated from above by a sequence $\tilde f_n$ of continuous functions. Those belong to $\cF_{\mu\nu}$ and the result follows after passing to the limit $n\to\ii$. Optimizing over $f\in\cF_{\mu\nu,R}$, we obtain the claim.
\end{proof}

The previous result justifies the use of the truncated problem $I_R(\mu,\nu)$ in place of the original problem $I(\mu,\nu)$. The following provides the existence of an optimizer for the truncated problem.

\begin{theorem}[Existence of an optimizer for $I_R(\mu,\nu)$]\label{thm:existence_f}
Let $R\geq2$. Let $\mu$ and $\nu$ be two a radial probability measures supported in the unit ball with $\mu(\{0\})=0$. There exists an optimal $f\in C^0([0,R],\R_+)$ solving the minimization problem $I_R(\mu,\nu)$ in~\eqref{eq:def_cI_mu_R}. This $f$ can be chosen Lipschitz-continuous and to satisfy the nonlinear equation
\begin{equation}
 f(a)=\max_{b\in[0,R]}\big\{\Psi_{\mu\nu}(a,b)-f(b)\big\},\qquad\forall a\in[0,R].
 \label{eq:nonlinear_f}
\end{equation}
In particular, we deduce that
$$0\leq f\leq g \qquad\text{on $[2,R]$,}$$
where $g$ is the function in Lemma~\ref{lem:G}.
\end{theorem}

We emphasize that the last bound holds only on $[2,R]$. We expect that $f$ is bounded independently of $R$ on $[0,2]$ (not necessarily by $g$), see Remark~\ref{rmk:equivalent_infinity} below.

\begin{proof}
The problems $I(\mu,\nu)$ and $I_R(\mu,\nu)$ take the same form as the dual problem in the theory of optimal transport~\cite{Villani-09,Santambrogio-15}. Our proof uses this analogy and we divide it into several steps. We assume throughout that $\nu\neq\delta_0$. The proof when $\nu=\delta_0$ works similarly, using~\eqref{eq:formula_Psi_delta}.

\medskip

\noindent\textbf{Step 1: \textit{a priori} bounds on $f$.} We can restrict the minimization in~\eqref{eq:def_cI_mu_R} to all the $f\in \cF_{\mu\nu,R}$ so that
\begin{equation}
\int_{B_R}\frac{f(|z|)}{|z|^7}\rd z\leq I_R(\mu,\nu)+1\leq I(\mu,\nu)+1\leq 2^{\frac43}12\pi+1\leq 100,
\label{eq:restrict_f}
\end{equation}
where the bound on $I(\mu,\nu)$ is by Lemma~\ref{lem:simple_estim_I_mu}. Taking first $a=b$ we obtain using~\eqref{eq:lower_bd_Psi}
$$f(a)\geq \frac{\Psi_{\mu\nu}(a,a)}{2}\geq a^6\frac{1-aV_\mu(ae_1)-aV_\nu(ae_1)}{2}.$$
Since $aV_\mu(ae_1)\to0$ and $aV_\nu(ae_1)\to\nu(\{0\})<1$ when $a\to0$ by~\eqref{eq:limit_V_mu}, we deduce that
$$f(a)\geq \frac{\Psi_{\mu\nu}(a,a)}{2}\geq a^6\frac{1-\nu(\{0\})}{4}$$
in a neighborhood $[0,\eps_0]$ of the origin, where $\eps_0<1$ depends only on $\mu$ and $\nu$.
On the other hand, averaging the constraint on $f$ over $b$ in a small interval $(0,\eps)$. We obtain
\begin{align*}
f(a)&\geq \frac1{\eps}\int_0^\eps\Psi_{\mu\nu}(a,b)\rd b-\frac1{\eps}\int_0^\eps f(b)\,\rd b\\
&\geq \frac1{\eps}\int_0^\eps\Psi_{\mu\nu}(a,b)\rd b-100\,\eps^6\\
&\geq a^3 \frac1{\eps}\int_0^\eps b^3\big(1-bV_\mu(be_1)-bV_\nu(be_1)\big)\rd b-100\,\eps^6
\end{align*}
for all $a\in[0,R]$. For $\eps\leq\eps_0$, we thus have
$$f(a)\geq a^3 \frac{\eps^3(1-\nu(\{0\})}{16}-100\,\eps^6.$$
Hence we have $f\geq a^3 \frac{\eps^3}{32}$ for $a\in[\eps_0,R]$ if we pick
$100\,\eps^3 = {\eps_0^3}/{32}.$
This proves that $f(a)\geq \kappa \max(a^6,a^3)$ for a small enough constant $\kappa\leq 1$ depending only on $\mu$ and $\nu$. So far the argument works the same if $R=+\ii$. On the compact interval $[0,R]$ we can simplify this to $f(a)\geq \kappa R^{-3}a^6$.

Consider next the new function $\tilde{f}(a)=\min(f(a),Ca^6)$ with $C=R^3/\kappa\geq8$. We claim that $\tilde f\in\cF_{\mu\nu,R}$.
Since
$$\Psi_{\mu\nu}(a,b)\leq 2a^3b^3\leq a^6+b^6\leq C(a^6+b^6),$$
we only have to consider the case where, say, $\tilde f(a)=f(a)$ and $\tilde f(b)=Cb^6$. In this case we write
$$\Psi_{\mu\nu}(a,b)\leq 2a^3b^3\leq \kappa R^{-3} a^6+\frac{R^3b^6}{\kappa}\leq f(a)+Cb^6=\tilde f(a)+\tilde f(b).$$
Thus the new function $\tilde f$ belongs to $\cF_{\mu\nu,R}$ as claimed. Since $\tilde f\leq f$ it also satisfies the constraint~\eqref{eq:restrict_f} and even gives a smaller integral. Thus in our minimization $I_R(\mu,\nu)$ we can always replace $f$ by $\tilde f$. In other words, we can always work with functions $f\in\cF_{\mu\nu,R}$ satisfying the additional condition
\begin{equation}
 \frac{\kappa}{R^3} a^6\leq f(a)\leq \frac{R^{3}}{\kappa}a^6
 \label{eq:info_f_1}
\end{equation}
where we recall that $\kappa$ only depends on $\mu$ and $\nu$.

\medskip

\noindent\textbf{Step 2: iterating the fixed point twice.} Let now $f\in\cF_{\mu\nu,R}$ satisfy the additional constraints~\eqref{eq:restrict_f} and~\eqref{eq:info_f_1}.   We define the two functions
$$f_1(a):=\max_{b\in[0,R]}\left\{\Psi_{\mu\nu}(a,b)-f(b)\right\},\qquad f_2(b):=\max_{a\in[0,R]}\left\{\Psi_{\mu\nu}(a,b)-g(a)\right\}.$$
Those would be denoted as $f^{\Psi_{\mu\nu}}$ and $f^{\Psi_{\mu\nu}\Psi_{\mu\nu}}$ ($\Psi_{\mu\nu}$-transforms) in the theory of optimal transport. Taking $b=0$ we find that $f_1\geq0$ since $\Psi_{\mu\nu}(a,0)=0$ and $f(0)=0$. Using that
$$\Psi_{\mu\nu}(a,b)\leq f(a)+f(b),\qquad \forall 0\leq a,b\leq R,$$
we get $f_1\leq f$. Now, let $b_a\in[0,R]$ be so that $f_1(a)=\Psi_{\mu\nu}(a,b_a)-f(b_a)$. We have
\begin{align*}
f_1(a)-f_1(a')&\leq \Psi_{\mu\nu}(a,b_a)-f(b_a)-\Psi_{\mu\nu}(a',b_a)+f(b_a)\\
&=\Psi_{\mu\nu}(a,b_a)-\Psi_{\mu\nu}(a',b_a) \leq CR^5|a-a'|
\end{align*}
by Lemma~\ref{lem:regularity_Psi}. Exchanging the role of $a$ and $a'$ proves that $f_1$ is Lipschitz-continuous on $[0,R]$:
\begin{equation}
 |f_1(a)-f_1(a')|\leq CR^5|a-a'|.
 \label{eq:g_Lipschitz}
\end{equation}
Since $0\leq f_1\leq f$, we conclude that $f_1(0)=0$ and can thus carry over the exact same argument for $f_2$. We obtain $0\leq f_2\leq f$ and that $f_2$ satisfies the same Lipschitz estimate~\eqref{eq:g_Lipschitz} as $f_1$.
By definition of $f_2$ we have
$$\Psi_ {\mu\nu}(a,b)\leq f_1(a)+f_2(b),\qquad\forall a,b\in[0,R].$$
Exchanging the roles of $a$ and $b$ we deduce that
$$\tilde f:=\frac{f_1+f_2}{2}$$
belongs to $\cF_{\mu\nu,R}$, satisfies the pointwise bound $0\leq \tilde f\leq f$ and is Lipschitz as in~\eqref{eq:g_Lipschitz}. Thus we can replace $f$ by $\tilde f$ and restrict our minimization problem $I_R(\mu,\nu)$ to the functions in $\cF_{\mu\nu,R}$ which satisfy~\eqref{eq:restrict_f},~\eqref{eq:info_f_1} and~\eqref{eq:g_Lipschitz}. This set is compact, by Ascoli's theorem. Since $f\mapsto\int_{B_R}|z|^{-7}f(|z|)\,\rd z$ is lower semi-continuous, we conclude that there exists an optimizer $f$ for $I_R(\mu,\nu)$, satisfying all the previous additional properties.

\medskip

\noindent\textbf{Step 3: properties of minimizers.}
For the previous minimizer $f$ we could go on and define the same functions $0\leq f_1,f_2\leq f$ as above. Those ought to have the same integral as $f$, by minimality. This proves that $f=f_1=f_2$, that is, $f$ solves the nonlinear equation~\eqref{eq:nonlinear_f}.

Finally, if $a\geq2$ then the maximum over $b$ in~\eqref{eq:nonlinear_f} must be attained for $b\leq2$ since $\Psi_{\mu\nu}(a,b)=0$ for $a,b\geq2$. Thus $b\leq a$ and
$$f(a)=\max_{b\leq a}\big\{\Psi_{\mu\nu}(a,b)-f(b)\big\}\leq g(a)$$
since $f(b)\geq\Psi_{\mu\nu}(b,b)_+/2$ and by definition of $g$ in Lemma~\ref{lem:G}.
\end{proof}

\begin{remark}\label{rmk:equivalent_infinity}\rm
From numerical simulations we expect that minimizers $f_R$ for $I_R(\mu,\nu)$ are in fact bounded independently of $R$ on $[0,2]$ (not necessarily by $g$). Should this be true, we could pass to the limit and get a minimizer for $I(\mu,\nu)$. Note that this minimizer would behave like $g$ at infinity in~\eqref{eq:G_infinity} since the same argument as above gives
\begin{equation}
g(a)-\max_{b\in [0,2]} \left\{f(b)-\frac{\Psi_{\mu\nu}(b,b)_+}{2}\right\}\leq f(a)\leq g(a),\qquad \forall a\geq2.
\label{eq:lower_bound_f_G}
\end{equation}
Introducing $I_R(\mu,\nu)$ is useful to get existence, but also helpful for the numerical implementation discussed in the next section.
\hfill$\diamond$
\end{remark}

\section{Numerical evaluation of the Lieb-Oxford constant} \label{sec:numerics}

In this section we explain how we have discretized and then approximately solved the minimization problem
\begin{equation}
 \inf_{\mu,\nu} I(\mu,\nu)D(\mu,\mu)^2
 \label{eq:remind_min}
\end{equation}
appearing on the right side of~\eqref{eq:main}.

\subsubsection*{\textbf{Discretizing $\mu$ and $\nu$}}
We reduce the problem to finite dimension by assuming that
\begin{itemize}
 \item $\mu$ and $\nu$ have compact support which, by scaling, can be taken in the unit ball;
 \item $\mu$ and $\nu$ are sums of uniform measures over spheres of radii $j/K$, $j=0,...,K$
\end{itemize}
for some $K\geq1$. For $\mu$ we do not allow a delta at the origin, but for $\nu$ a delta is permitted and corresponds to $j=0$. It is important to note that the previous approximation always yields an \emph{upper bound} on the full minimum in~\eqref{eq:remind_min}, hence on the best constant $c_{\rm LO}$. Even small values of $K$ can yield some information.

With the above approximation, the minimization problem~\eqref{eq:remind_min} is posed in dimension $2K+1$. Although computing $D(\mu,\mu)$ is easy and explicit, the main difficulty is to compute $I(\mu,\nu)$.

\subsubsection*{\textbf{Approximate computation of $I(\mu,\nu)$}}
Let $\mu$ and $\nu$ be sums of uniform measures over spheres of radii $j/K$, $j=1,...,K$, together with a $\delta_0$ for $\nu$:
\begin{equation}
\mu=\frac1K\sum_{j=1}^K\mu_j\, \sigma_{\frac{j}K},\quad \nu=\nu_0\delta_0+\frac1K\sum_{j=1}^K\nu_j\, \sigma_{\frac{j}K},\quad \frac1K\sum_{j=1}^K\mu_j=\nu_0+\frac1K\sum_{j=1}^K\nu_j=1,
 \label{eq:mu_discretized}
\end{equation}
where $\sigma_r$ is the normalized delta measure on the sphere of radius $r$. Then $\Psi_{\mu\nu}$ can be expressed using~\eqref{eq:formula_Psi_delta} as
\begin{multline}
\Psi_{\mu\nu}(a,b)=\frac1{K^8}\sum_{j,k=1}^Kj^3k^3\Psi_{\sigma\sigma}\left(\frac{Ka}{j},\frac{Kb}{k}\right)(\mu_j\nu_k+\nu_j\mu_k-\nu_j\nu_k)+2\nu_0^2a^3b^3\\
-\frac{a^3b^3\nu_0}{K}\sum_{j=1}^K(\mu_j+\nu_j)\left(\min\left(\frac{aK}j,1\right)+\min\left(\frac{bK}j,1\right)\right).
\label{eq:Psi_mu_average_spheres2}
\end{multline}
In order to solve the minimization problem $I(\mu,\nu)$ approximately, we first choose an $R\geq2$ and use~\eqref{eq:estimate_I_R} to infer
$$I(\mu,\nu)\leq  I_R(\mu,\nu)+4\pi\int_R^\ii\frac{g(r)}{r^5}\rd r$$
where $g$ is the function from Lemma~\ref{lem:G}. In our simulations we have observed that $g(r)/r^3$ was always decreasing to its limit at infinity. The function started to decrease way before reaching the values of $R$ we took. Hence we always used the simple bound
\begin{equation}
 I(\mu,\nu)\leq  I_R(\mu,\nu)+4\pi \frac{g(R)}{R^4}.
 \label{eq:upper_I_mu_numerics}
\end{equation}
Although we have no proof that this bound is valid, we believe there is no approximation here.
In practice we took $10\leq R\leq 40$, which was sufficient to attain the desired precision. Of course, this all requires to compute an approximation of the function $g$, which we now discuss.

We discretized the minimization problem $I_R(\mu,\nu)$ in radial coordinates on a grid $(\frac{m}M)_{0\leq m\leq MR-1}$ containing $M$ points per unit length. First we computed the $(RM)\times (RM)$ matrix
\begin{equation}
 \psi_{\ell m}=\Psi_{\mu\nu}\left(\frac{\ell}{M},\frac{m}{M}\right),\qquad 0\leq \ell,m\leq RM-1
 \label{eq:matrix_Psi_exact}
\end{equation}
where $\Psi_{\mu\nu}$ is itself given by~\eqref{eq:Psi_mu_average_spheres2}. Computing this matrix scales like $R^2K^2M^2$ and is thus numerically very demanding. In practice, we do this in parallel on a cluster containing 40 GPUs. When we need to compute $\psi$ many times for different $\mu$ and $\nu$'s, we proceed differently and instead first construct and store the tensor
\begin{equation}
 \cT_{jk\ell m}=\left(\frac{jk}{\ell m}\right)^3\Psi_{\sigma\sigma}\left(\frac{K\ell}{M j},\frac{Km}{M k}\right),
 \label{eq:tensor}
\end{equation}
which is then used to compute $\psi$. Storing $\cT$ requires a lot of memory and was possible only for $K\sim M\simeq 50$ and $R\simeq 10$.

With the matrix $\psi$ at hand, we can replace $I_R(\mu,\nu)$ by its discretization over the grid
\begin{equation}
 I_{R,M}(\mu,\nu):=\inf_{F\in\cF_{\psi,R,M}}4\pi M^4\sum_{m=1}^{RN-1}\frac{F_m}{m^5}
 \label{eq:def_I_discretized}
\end{equation}
where
$$\cF_{\psi,R,M}=\bigg\{F=(F_m)_{m=0}^{RM-1}\ :\  \psi_{\ell m}\leq F_\ell+F_m,\quad \forall 0\leq \ell,m\leq RM-1,\ F_0=0\bigg\}.$$
Similarly, we introduce the discretization of the function $g$ from Lemma~\ref{lem:G}
\begin{equation}
G_m=\max_{0\leq \ell\leq m}\left\{\psi_{\ell m}-\frac{(\psi_{\ell\ell})_+}{2}\right\},\qquad \text{for $0\leq m\leq RM-1$.}
\label{eq:def_G_discretized}
\end{equation}
Our discretized approximation of $I(\mu,\nu)$ is, thus,
\begin{equation}
 I_{R,M}(\mu,\nu)+4\pi \frac{G_{RM-1}}{R^4}.
 \label{eq:approximate_I}
\end{equation}
We should mention here that this is \emph{a priori} not an upper bound to the true value $I(\mu,\nu)$. Thanks to Theorem~\ref{thm:existence_f} we know that minimizers are Lipschitz, which can be used to prove that
$$I_{R,M}(\mu,\nu)\geq I_{R}(\mu,\nu)-\frac{C_R}{M}.$$
If we use the estimate~\eqref{eq:g_Lipschitz} from the proof of Theorem~\ref{thm:existence_f} we get a constant $C_R$ behaving like $R^5$. In reality the obtained function $f$ always looks very smooth, independently of $R$, which suggests that $C_R$ does not depend on $R$.

It remains to explain how we solved the minimization problem $I_{R,M}(\mu,\nu)$. This is the minimization of a linear functional under inequality constraints (linear programming) which takes the same form as the discrete dual optimal transport problem~\cite{PeyCut-19}. Let us define the $\psi$-transform of \emph{any} non-negative vector $F$ with $F_0=0$ by
$$F^\psi_m:=\max_{0\leq \ell\leq RM-1}\left\{\psi_{\ell m}-F_\ell\right\}.$$
Then we have by definition $0\leq F^\psi\leq F$ (componentwise) as well as
$$\psi_{\ell m}\leq F_\ell+F^\psi_m,\qquad \forall 0\leq \ell,m\leq RM-1.$$
By symmetry we conclude that $(F+F^\psi)/2\in\cF_{\psi,R,M}$. From this we deduce that we can remove the constraint at the expense of adding a $\psi$-transform:
\begin{equation}
 I_{R,M}(\mu,\nu)=\inf_{F\geq0}4\pi M^4\sum_{m=1}^{RN-1}\frac{F_m+F_m^\psi}{2m^5}.
 \label{eq:def_I_discretized_convex}
\end{equation}
This way we obtain a convex nonlinear optimization problem in $F$.

If we give ourselves a vector $F^{(0)}$ in $\cF_{\mu,R,N}$, we have a simple way of decreasing the sum on the right side of~\eqref{eq:def_I_discretized_convex}. We just define inductively
\begin{equation}
F^{(n+1)}=\frac{F^{(n)}+(F^{(n)})^\psi}{2},
 \label{eq:iteration_f}
\end{equation}
which also belongs to $\cF_{\psi,R,M}$. By construction, we obtain a decreasing sequence $F^{(n+1)}\leq F^{(n)}$ and thus the sum decreases and converges. We can stop whenever the relative error is less than a prescribed $\varepsilon$ (taken equal to $\varepsilon=10^{-6}$ in our case). If we start with either $F^{(0)}_m=\max\big(2(m/M)^6,8(m/M)^3\big)$ by Lemma~\ref{lem:simple_estim_I_mu}, or $F^{(0)}=G$ (the discretization of the function $g$ from Lemma~\ref{lem:G}), the iterative algorithm stops after a few iterations for $100\leq M\leq 1000$.

This method allows us to easily find an extreme point $F=\lim_{n\to\ii}F^{(n)}$ of the cone $\cF_{\psi,R,M}$, which has a sum in~\eqref{eq:def_I_discretized} lower than our initial vector $F^{(0)}$. In principle, this point has no reason of being a global minimizer for $I_{R,M}(\mu,\nu)$. Nevertheless, when we tested this method in our situation using standard algorithms on~\eqref{eq:def_I_discretized_convex}, we could never beat the extreme point obtained by setting $F^{(0)}=G$.  In order to spare computational time, we thus used this $F$ as an (upper) approximation to $I_{R,M}(\mu,\nu)$ everywhere. The limiting $F$ was found to be always very close to the vector $G$, and sometimes even exactly equal (in which case the algorithm~\eqref{eq:iteration_f} stops at the first step).

\subsubsection*{\textbf{Examples: balls and spheres}}
We tested the computation of $I(\mu,\nu)$ for $\mu$ and $\nu$ equal to the uniform measure of either the sphere or of the ball. We can also allow $\nu$ to be a delta at the origin to compare with the Lieb-Oxford original approach~\cite{LieOxf-80}. In all these cases we can compute the function $\Psi_{\mu\mu}$ exactly. For instance, for $\mu=\nu=\cB:=\frac{3}{4\pi}\1_{B_1}$, a tedious calculation provides
\begin{multline}
 \Psi_{\cB,\cB}(a,b)=\frac{(a+b-ab)_+^4}{160}\left(a^2b^2-5a^2-5b^2+4ab^2+4ba^2+20ab\right)\\
 -\frac{(|a-b|-ab)_+^4}{160}\left(a^2b^2-5a^2-5b^2+4|a-b|ab-20ab\right).
 \label{eq:Psi_ball}
\end{multline}
In principle, there are explicit formulas for $\mu$ any polynomial in $r$. In Table~\ref{tab:balls_spheres} we report the values we found for the LO constant when $\mu$ and $\nu$ are either $\cB$ or $\sigma$. The best constant is obtained for $\mu=\nu=\cB$ and it is already much better (slightly above $1.60$) than all previously known results. Note the slight improvement $1.6583$ in the Lieb-Oxford case $\mu=\cB$ and $\nu=\delta_0$, compared to the original value $1.68$ obtained in~\cite{LieOxf-80}. This is due to the optimization of $I(\mu,\delta_0)$ instead of using the bound~\eqref{eq:LO_bound}.

\begin{table}[t]
\begin{tabular}{r|l|l|l}
 & $\nu=\delta_0$ & $\nu=\sigma$ &$\nu=\cB$\\
\hline
$\mu=\sigma$ & 1.7829 &1.7019 &1.7172\\
$\mu=\cB$ & 1.6583&1.6444 &1.6044\\
\end{tabular}

\medskip

\caption{\footnotesize Value of the LO constant found for $\mu$ and $\nu$ being either $\cB=(3/4\pi)\1_{B_1}$ (uniform measure of the unit ball), or $\sigma$ (uniform measure of the unit sphere), or the Dirac delta $\delta_0$. We use here the exact formula for $\Psi_{\mu\nu}$ and discretize $I(\mu,\nu)$ using $M=500$ and $R=30$. Numbers are all rounded up to the fourth decimal.\label{tab:balls_spheres}}
\end{table}

In Figure~\ref{fig:example_f} we display the optimal vector $F=\lim_{n\to\ii}F^{(n)}$ which we found by applying the algorithm described before for $F^{(0)}=G$ and $\mu=\nu$, with $R=20$ and $M=1000$. In fact in the picture we rather draw $F(r)/r^3$, which is almost constant for large $r$. The solution $F$ was found to coincide with the vector $G$ in these two cases.

\begin{figure}[t]
\centering

\begin{tabular}{cc}
\includegraphics[width=6.4cm]{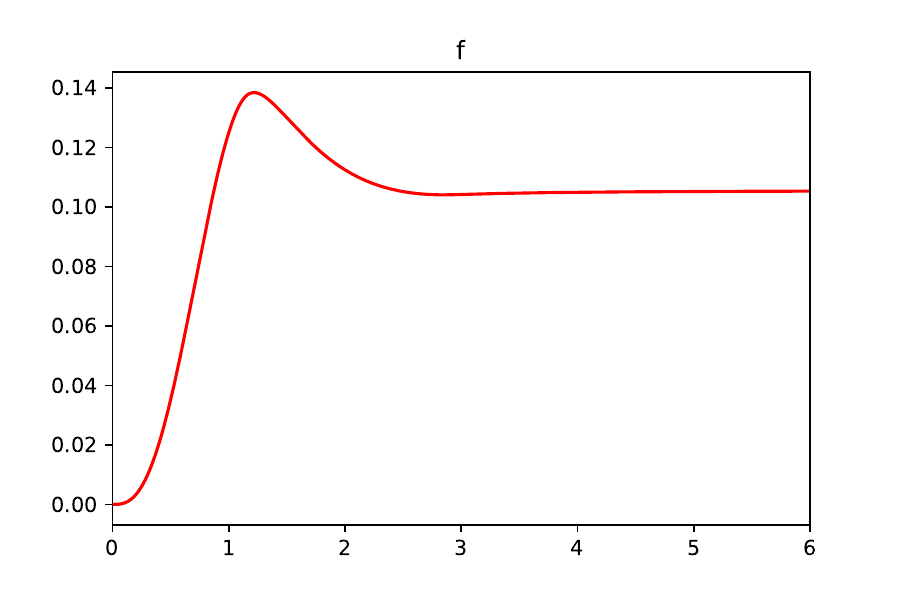}&\includegraphics[width=6.4cm]{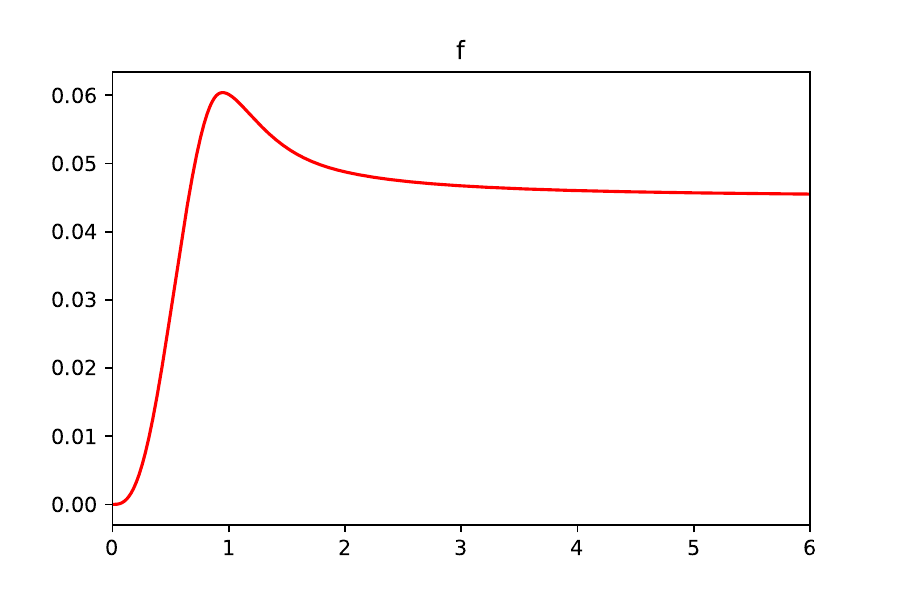} \\
\small($\mu=\nu=\sigma$)&\small ($\mu=\nu=\cB$)
\end{tabular}

\caption{\footnotesize Plot of the numerical approximations of the function $r\mapsto r^{-3}f(r)$, for $\mu=\nu=\sigma$ (delta measure of the sphere) on the left and $\mu=\nu=\cB=(3/4\pi)\1_{B_1}$ (uniform measure of the unit ball) on the right. In these two cases, the solution $f$ seems to coincide with the function $g$ from Lemma~\ref{lem:G}. We used here the exact formula for $\Psi_{\mu\nu}$, which we discretized into the matrix $\psi$ by~\eqref{eq:matrix_Psi_exact} on a grid with $M=1000$ points per unit length and $R=20$ to compute the approximation $F$ to the function $f$. The shapes in the two cases are very similar but the ball gives a much lower function, hence a much better approximation $c_{\rm LO}\leq 1.6044$.\label{fig:example_f}}
\end{figure}

The exact knowledge of $\Psi_{\mu\nu}$ allows us to investigate separately the efficiency of the discretization of the measures using concentric spheres with the parameter $K$, and that of the computation of $I(\mu,\nu)$ with the parameters $M$ and $R$.

\medskip

\paragraph*{\textit{Convergence in $M$ and $R$}}
In Table~\ref{tab:M_R_ball} we report the LO constant obtained for $\mu=\nu=\cB$, using the exact formula of $\Psi_{\cB\cB}$, in terms of the two parameters $M$ and $R$. Essentially, the numbers are decreasing with $R$, confirming the validity of our upper approximation~\eqref{eq:approximate_I}, and increasing in $M$. For the displayed values of $M$ we already obtain a reliable approximation to the order $10^{-4}$.

\begin{table}[t]
\begin{tabular}{r|l|l|l|l}
 & $R=10$ & $R=20$ & $R=30$ & $R=40$\\
\hline
$M=100$ & 1.604358& 1.604317&1.604312 & 1.604311\\
$M=200$ & 1.604373& 1.604334&1.604330 &1.604329\\
$M=300$ & 1.604375&1.604337& 1.604334&1.604333\\
 $M=500$ &1.604377&1.604340& 1.604336& \\
 $M=1000$ &1.604377 &1.604340&  &\\
\end{tabular}

\medskip

\caption{\footnotesize Value of the LO constant found for $\mu=\nu=\cB$ (uniform measure of the unit ball), depending on the discretization parameters $M$ and $R$ for the approximation of $I(\mu,\nu)$. Here the exact formula~\eqref{eq:Psi_ball} of $\Psi_{\cB\cB}$ is used in the computation.\label{tab:M_R_ball}}
\end{table}

\medskip

\paragraph*{\textit{Convergence in $K$.}}
In order to test the accuracy of the discretization~\eqref{eq:mu_discretized}, we report in Table~\ref{tab:value_CLO_N} the values we obtain with the exact same method, if we instead discretize the uniform ball as a combination of concentric spheres. We conclude that if we are only interested in the first few digits we can safely work with $K\leq 100$.

\begin{table}[t]
\begin{tabular}{l|l|l|l|l|l}
$K$&10&20&50&100&$\ii$\\
\hline
&1.606748&1.604961&1.604440&1.604364&1.604337\\
\end{tabular}

\medskip

\caption{\footnotesize Value of the Lieb-Oxford constant found for different values of the number $K$ of concentric spheres used to represent the uniform ball $\cB$ for $\mu=\nu$. Here we used $R=20$ and $M=300$. The value found for these parameters using the exact formula~\eqref{eq:Psi_ball} of $\Psi_{\cB\cB}$ is indicated in the column $K=\ii$.
\label{tab:value_CLO_N}}
\end{table}

\subsubsection*{\textbf{Optimizing over $\mu$ and $\nu$}}
In order to push the constant further down, we have to optimize over $\mu$ and $\nu$.
The algorithm requires many evaluations of the function $I_{R,M}(\mu,\nu)$ as well as to re-compute the matrix $\psi$ for each new measures $\mu,\nu$. This takes a lot of time. Storing the tensor~\eqref{eq:tensor} in memory seems the best option, but this limited the value of the number $K$ of spheres, the number $M$ of grid points and the cut-off $R$. We thus only used the minimization algorithm for $K=50$, $M=100$ and $R=10$. We used the standard Broyden-Fletcher-Goldfarb-Shanno (BFGS) algorithm to optimize the coefficients $\mu_j,\nu_j$ in~\eqref{eq:mu_discretized}. We employed the parallel implementation of BFGS from the Python package \emph{optimParallel}~\cite{optimParallel,optimParallel-R}, on a cluster of 40 GPUs.

The best solution we obtained after testing many random initial states is displayed in Figure~\ref{fig:optimal_mu}. It has $c_{\rm LO}\leq1.5765$ for $K=50$, $M=100$ and $R=10$. It is surprising that $\mu$ ends up being smooth, with a support slightly smaller than the unit ball, whereas $\nu$ is less smooth and close (but not exactly equal) to three deltas on spheres, including one on the unit sphere.

This solution $(\mu,\nu)$ with $50$ concentric spheres can now be used to provide an upper bound on $c_{\rm LO}$, by computing a sufficiently good approximation of $I(\mu,\nu)$, that is, varying only $M$ and $R$. For that $\psi$ we found $G^\psi\neq G$ and it was thus necessary to apply the iterative algorithm~\eqref{eq:iteration_f} to improve the upper bound. We display the result for several values of $M$ and $R$ in Table~\ref{tab:M_R}. Given these results we can safely assert that
\begin{equation}
\boxed{c_{\rm LO}\leq 1.5765.}
\label{eq:LO_final}
\end{equation}

\begin{figure}[t]
\centering
\includegraphics[width=6.5cm]{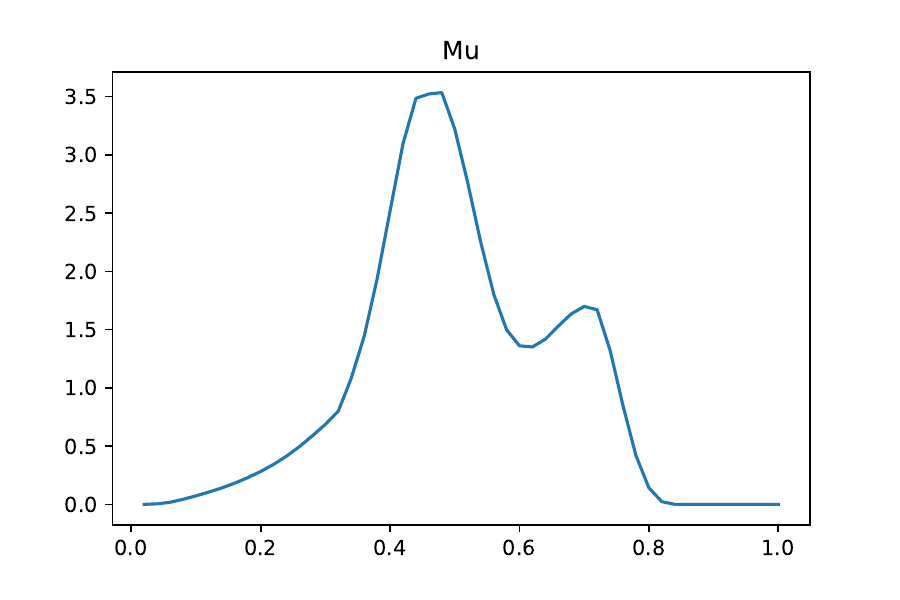}\includegraphics[width=6.5cm]{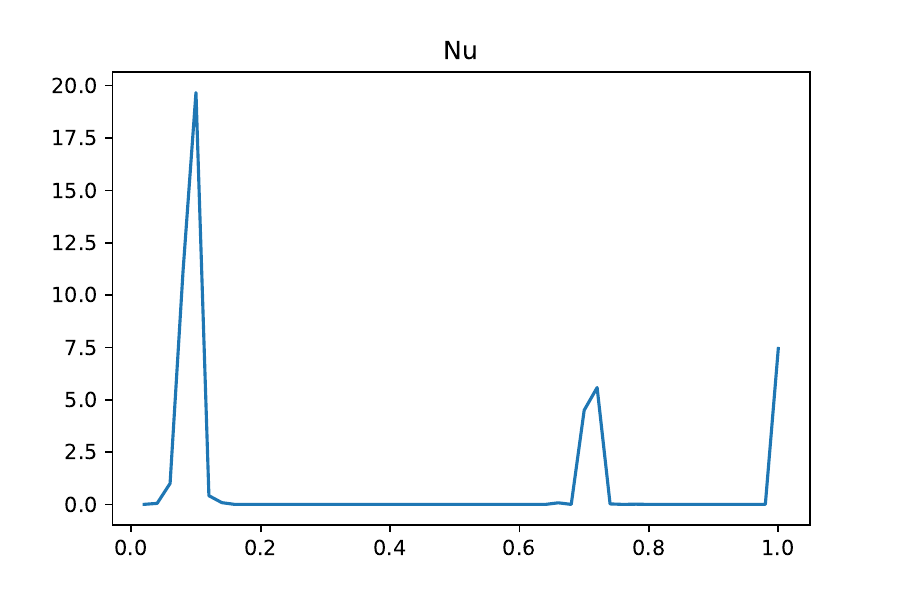}

\caption{\footnotesize Plot of the two radial measures $r\mapsto r^2\mu(r)$ (left) and $r\mapsto r^2\nu(r)$ (right) found by the BFGS algorithm for $K=50$, $M=100$, $R=10$. We obtain the upper bound $c_{\rm LO}\leq 1.5765$ claimed in~\eqref{eq:LO_final}. \label{fig:optimal_mu}}
\end{figure}

\begin{table}[t]
\begin{tabular}{r|l|l|l|l}
 & $R=10$ & $R=20$ & $R=30$ & $R=40$\\
\hline
$M=100$ & 1.576395&1.576364 &1.576360 & 1.576359\\
$M=200$ & 1.576441&1.576410 & 1.576406 &1.576405\\
$M=300$ & 1.576446&1.576417 & 1.576413 &\\
$M=400$ & 1.576446 &1.576419 &  &\\
\end{tabular}

\medskip

\caption{\footnotesize Value of the LO constant found for the two measures $\mu$ and $\nu$ in Figure~\ref{fig:optimal_mu} (which have $K=50$ concentric spheres), depending on the discretization parameters $M$ and $R$ for the approximation of $I(\mu,\nu)$.\label{tab:M_R}}
\end{table}


\newcommand{\etalchar}[1]{$^{#1}$}

\end{document}